\documentclass[11pt]{article}

  \usepackage{amsthm}
  \usepackage{amssymb}
  \usepackage{amsmath}
  \usepackage{amscd}
  \usepackage{eucal}
  \usepackage{mathrsfs}
  \usepackage{upgreek}
   \usepackage{dsfont}
      \usepackage{yfonts}
  \usepackage[T1]{fontenc}
  \usepackage{aurical}
  \usepackage{bbm}
  \usepackage{geometry}
\usepackage{array}
\usepackage{booktabs}
\newcommand{\otoprule}{\midrule[\heavyrulewidth]}   
\newcommand{\Tens}{\mathfrak{Tens}}
\newcommand{\CE}{\mathfrak{CE}}  
\newcommand{\E}{\mathfrak{E}}
\newcommand{\B}{\mathfrak{B}}
\newcommand{\W}{\mathfrak{W}}
\newcommand{\D}{\mathfrak{D}}
\newcommand{\F}{\mathfrak{F}}
\newcommand{\fG}{\mathfrak{G}}
\newcommand{\V}{\mathfrak{V}}
\newcommand{\BV}{\mathfrak{BV}}
\newcommand{\Nm}{\mathfrak{Nm}}
\newcommand{\BVn}{\BV_{\textrm{n}}}
\newcommand{\frakg}{\mathfrak{g}_c}
\newcommand{\frakgo}{\mathfrak{g}}
\newcommand{\Xo}{\mathfrak{X}}
\newcommand{\1}{\mathds{1}}

\newcommand{\A}{\mathcal{A}} 
\newcommand{\G}{\mathcal{G}} 
\newcommand{\al}{\alpha}
\newcommand{\bet}{\beta}

\newcommand{\la}{\lambda}
\newcommand{\La}{\Lambda}

\newcommand{\ph}{\varphi}

\newcommand{\Nat}{\mathrm{Nat}}

\newcommand{\Bndl}{\mathrm{\mathbf{Bndl}}}
\newcommand{\Loc}{\mathrm{\mathbf{Loc}}} 
\newcommand{\Vect}{\mathrm{\mathbf{Vec}}} 
\newcommand{\dgA}{\mathrm{\mathbf{dgA}}}
\newcommand{\PgAlg}{\mathrm{\mathbf{PgAlg}}}
\newcommand{\gt}[2]{\tilde{\textfrak{g}}^{#1#2}}

\newcommand{\sst}[1]{\scriptscriptstyle{#1}}

\newcommand{\hinv}{*^{\!\sst{-\!1}}}
\newcommand{\be}{\begin{equation}}
\newcommand{\ee}{\end{equation}}

\newcommand{\pa}{\partial} 

\newcommand{\NN}{\mathbb{N}} 
\newcommand{\RR}{\mathbb{R}} 
\newcommand{\CC}{\mathbb{C}} 
\newcommand{\supp}{\mathrm{supp}}
\newcommand{\dvol}{d\mathrm{vol}_M}
\newcommand{\WF}{\mathrm{WF}} 
\newcommand{\loc}{\mathrm{loc}}
\newcommand{\ml}{\mathrm{ml}}
\newcommand{\inv}{\mathrm{inv}}

\newcommand{\pg}{\mathrm{pg}}

\newcommand{\af}{\mathrm{af}}
\newcommand{\ta}{\mathrm{ta}}
\newcommand{\gh}{\mathrm{gh}}

\newcommand{\mc}{\mathrm{mc}}
\newcommand{\ex}{\mathrm{ext}}
\newcommand{\Ci}{\mathcal{C}^\infty}

\newcommand{\Bcal}{\mathcal {B}}

\newcommand{\Dcal}{\mathcal{D}}
\newcommand{\Ecal}{\mathcal{E}}

\newcommand{\Diff}{\mathrm{Diff}}
\newcommand{\im}{\mathrm{Im}}
\newcommand{\ke}{\mathrm{Ker}}
\newcommand{\Der}{\mathrm{Der}}

\newcommand{\tr}{\mathrm{tr}}

\newcommand{\alte}{\mathrm{alt}}
\newcommand{\X}{\mathfrak{X}}

 \author{\null\\ Klaus Fredenhagen, Katarzyna Rejzner \\
  \null\\
  \null\\
        \small{ II Inst. f. Theoretische Physik, Universit\"at Hamburg,}\\
    \small{Luruper Chaussee 149, D-22761 Hamburg, Germany}\\ 
\small{\texttt{klaus.fredenhagen@desy.de,katarzyna.rejzner@desy.de}}}

  \title{Batalin-Vilkovisky formalism in the functional approach to classical field theory}

\begin{document}

 \maketitle

\begin{abstract}
We develop the Batalin-Vilkovisky formalism for classical field theory on generic globally hyperbolic spacetimes. A crucial aspect of our treatment is the incorporation of the principle of local covariance which amounts to formulate the theory without reference to a distinguished spacetime. In particular, this allows a homological construction of the Poisson algebra of observables in classical gravity. Our methods heavily rely on the differential geometry of configuration spaces of classical fields.
\end{abstract}

  \theoremstyle{plain}
  \newtheorem{df}{Definition}[section]
  \newtheorem{thm}[df]{Theorem}
  \newtheorem{prop}[df]{Proposition}
  \newtheorem{cor}[df]{Corollary}
  \newtheorem{lemma}[df]{Lemma}
  
  \theoremstyle{plain}
  \newtheorem*{Main}{Main Theorem}
  \newtheorem*{MainT}{Main Technical Theorem}

  \theoremstyle{definition}
  \newtheorem{oss}[df]{Remark}

 \theoremstyle{definition}
  \newtheorem{ass}{\underline{\textit{Assumption}}}[section]


\tableofcontents
\markboth{Contents}{Contents}
\section{Introduction}
\label{intro}
Classical field theory is an infinite dimensional generalization of classical mechanics, Therefore the transfer of structures like Lagrangians, Hamiltonians, Poisson brackets etc. 
involve problems of infinite dimensional analysis. To a certain extent these problems can be circumvented by exploiting the locality principle of classical field theory and using as variables finitely many derivatives of the field at a given point. A particular elegant way of doing this is multisymplectic analysis \cite{Kij,Got,CCI}.

A drawback of this approach is that the connection to quantum field theory is not easy (for an attempt see \cite{Kanachikov1,Kanachikov2}). Other approaches use the path integral formalism of quantum field theory. There it is extremely difficult to arrive at rigorous results. 

A new way of formulating perturbative quantum field theory which combines aspects of canonical quantization with aspects from the path integral formalism is perturbative  algebraic quantum field theory \cite{BDF,DF,DF02,DF04}. It inherits from the canonical approach the emphasis on algebraic structures and from the path integral formalism the insight that all possible field configurations have to be considered, not only those which are solutions of the field equations. This approach goes back to old ideas of Schwinger, St\"uckelberg and Bogoliubov, was made mathematically rigorous by the method of causal perturbation theory (Epstein-Glaser renormalization) and was mainly developed in order to treat quantum field theory on generic curved spacetimes in agreement with the principle of general covariance. The $\hbar=0$ limit then is classical field theory \cite{BDF} where the observables are functionals of generic field configurations and the Poisson bracket is a functional bi-differential operator.
The general method was mainly developed on the example of a scalar field.

A generalization to gauge theories was performed in several steps \cite{Dutsch:2004hu,DF02}, culminating in the work of Hollands \cite{H}
on quantization of Yang-Mills theories on curved backgrounds. There it was observed that the Batalin-Vilkovisky formalism \cite{Batalin:1977pb,Batalin:1981jr,Batalin:1983wj,Batalin:1983jr} is most suitable for the treatment of gauge theories.

The Batalin-Vilkovisky (BV) formalism (previously and independently introduced by Zinn-Justin \cite{Zinn}) generalizes the BRST approach \cite{BRST1,BRST2}. In the literature it is also called the antifield formalism. It is usually used as a covariant method to perform the gauge-fixing in quantum field theory, but was also applied to other problems like analyzing possible deformations of the action and anomalies. For a review on BRST see \cite{HennBar}, where the cohomological interpretation of the antifield formalism is especially stressed. It was already recognized in \cite{Henneaux:1992ig} (see also \cite{Henneaux:1989jq}) that the Koszul-Tate complex plays a very important role in the antifield formalism and the methods of homological algebra can be successfully applied. The book \cite{Henneaux:1992ig} shows also the geometrical interpretation of the BV complex in both the Hamiltonian and the Lagrangian setting, but the mathematically rigorous results are formulated only when the configuration space is finite dimensional. A generalization to infinite dimensional spaces is usually achieved in the jet space formalism \cite{HennBar}, together with the variational bicomplex \cite{Sard2,Sard,Sard1}. This formulation is sometimes more convenient for calculation but it is hard to see the general structure behind it. Furthermore it is usually applied only to the local functionals of the configuration fields and in quantum field theory one needs more singular objects (see for example \cite{BDF}).

Most of the mathematically rigorous versions of the BV formalism (see, e.g., \cite{Stasheff2,Stasheff,Froe}) typically use assumptions which are not satisfied by the examples of interest in physics, e.g. that the underlying spacetime is compact (this is never true for a globally hyperbolic spacetime), or that the field configurations of interest have compact support (this excludes solutions of hyperbolic field equations), or spacetime is even replaced by a finite set in order to make the path integral well defined. 
Moreover the formalism looks like a formal recipe where a lot of additional structure is introduced (ghosts, antifields,\dots) whose conceptual and mathematical status is often unclear. A more conscious approach, with a due attention to the infinite dimensional nature of the problem was proposed in the notes \cite{Cost2} but the formalism is still not complete.

In the present paper we attempt to provide a conceptually convincing and mathematically rigorous version of the BV formalism  
for classical field theory on globally hyperbolic spacetimes. Our construction is based on the principle of local covariance \cite{BFV}. According to this principle one does no longer try to construct a theory on a given spacetime. Instead all constructions have to be performed simultaneously on all spacetimes of a given class in a coherent way. This amounts to a functorial relation between spacetimes and the associated Poisson algebras of observables. We hope that our formulation improves the understanding of classical field theory (see also \cite{BFR}), but our main goal is to use it as a basis for quantum field theory, in particular for quantum gravity as proposed in \cite{F,BF1}.
\section{The BV formalism for the scalar field}\label{scalar}
Before entering the more complicated realm of gauge theories, we want to illustrate the basic structures on the example of a scalar field. 
     
Let $\E(M)$ be the space of smooth functions on the smooth manifold $M$, equipped with the locally convex topology of convergence of all derivatives, uniformly on compact sets. Let $\chi:M\to N$ be an embedding and let 
$\chi^*$ be the pullback, $\chi^*\ph=\ph\circ\chi$ for $\ph\in\E(M)$. Then $\E$ can be understood as a contravariant functor
from the category $\Loc$ of  time-oriented globally hyperbolic spacetimes and with causal isometric embeddings
 as morphisms to the category $\Vect$ of locally convex vector spaces (lcvs) with continuous injective linear maps as arrows. The embedding $\chi$ is associated to the pullback $\chi^*$,  $\E \chi=\chi^*$. We interpret $\E(M)$ as the space of field configurations. 
\footnote{In view of the expected properties of the path integral one might conjecture that field configurations which are not smooth will play an important role. Fortunately this is not the case in a perturbative framework. Since our observables are functions of field configurations, an extension to more general configurations would mean that we had to distinguish functions which coincide on smooth configurations.}

Another functor between these categories is the functor which associates to a manifold the space $\D(M)$ of compactly supported test functions with its standard topology. This functor is covariant with respect to push forwards $\D\chi=\chi_*$,
\[
\chi_*f(x)=\left\{
\begin{array}{ccc}
f(\chi^{-1}(x)) & , & x\in\chi(M),\\
0                    & , & \text{else}.
\end{array} \right.
\]
The observables of the theory are functions $F:\E(M)\to\RR$
.They are supposed to be smooth in the following sense:
For each $n\in\NN$ and each $\ph\in\E(M)$ there exists a distributional density $F^{(n)}(\ph)$ with compact support on $M^n$.
$F^{(n)}(\ph)$ is symmetric under permutations of the $n$ arguments from $M$ and satisfies
\[
\frac{d^n}{d\la^n}F(\ph+\la\psi)\restriction_{\la=0}=\langle F^{(n)}(\ph),\psi^{\otimes n}\rangle \ .
\]
Moreover, the map $\E(M)\times\E(M)\ni(\ph,\psi)\mapsto \langle F^{(n)}(\ph),\psi^{\otimes n}\rangle$ is required to be continuous. Since $\E(M)$ is a nuclear Fr\'echet space, this notion coincides with the standard definition of smooth maps between locally convex vector spaces (see the Appendix \ref{iddg} for details).

The support of a function $F$ on configuration space is defined as the set of points $x\in M$ such that $F$ depends on the field configuration in any neighbourhood of $x$
\begin{align}\label{support}
\supp\, F=\{ & x\in M|\forall \text{ neighbourhoods }U\text{ of }x\ \exists \ph,\psi\in\E(M), \supp\,\psi\subset U 
\\ & \text{ such that }F(\ph+\psi)\not= F(\ph)\}\ .\nonumber
\end{align}
A very important property of functionals in the context of classical field theory is \textit{locality}. According to the standard definition it means that the functional $F$ is of the form:
\[
F(\ph)=\int\limits_M f(j_x(\ph))\,d\mu(x)\,,
\]
where $f$ is a function on the jet space over M and $j_x(\ph)=(x,\ph(x),\pa\ph(x),\dots)$ is the jet of $\ph$ at the point $x$. It was already recognized in \cite{DF04,BDF,BFR} in the context of perturbative algebraic quantum field theory that the property of locality can be reformulated using the notion of \textit{additivity} of a functional. The concept itself dates back to the works of Chac\'on and Friedman \cite{ChF} and even to earlier ones as seen from the survey of Rao \cite{Rao}. We say that $F$ is additive if for all fields $\varphi_1,\varphi_2,\varphi_3\in\E(M)$ such that $\textrm{supp}(\varphi_1)\cap\textrm{supp}(\varphi_3)=\varnothing$ we have:
\be\label{add}
F(\varphi_1+\varphi_2+\varphi_3)=F(\varphi_1+\varphi_2)-F(\varphi_2)+F(\varphi_2+\varphi_3)\,.
\ee
One can show that a smooth compactly supported functional is local if it is additive and
the wave front sets of its derivatives are orthogonal to the tangent bundles of the thin diagonals $\Delta^k(M)\doteq\left\{(x,\ldots,x)\in M^k:x\in M\right\}$, considered as subsets of the tangent bundles of $M^k$:
\be
\textrm{WF}(F^{(k)}(\ph))\perp T\Delta^k(M)\,.\label{WFloc}
\ee
In particular $F^{(1)}(\ph)$ is a smooth section for each fixed  $\ph$. This property turns out to be crucial in the context of the BV formalism. We shall come back to it later on. 

The space of compactly supported smooth local functions $F:\E(M)\to\RR$ is denoted by $\F_\loc(M)$. The algebraic completion of $\F_\loc(M)$ with respect to the pointwise product
\be\label{prod}
F\cdot G(\ph)=F(\ph)G(\ph) \,,
\ee
is a commutative algebra $\F(M)$ consisting of finite sums of finite products of local functionals. We call it \textit{the algebra of multilocal functionals}.
 $\F$ becomes a (covariant) functor by setting $\F\chi(F)=F\circ \E\chi$, i.e. $\F\chi(F)(\ph)=F(\ph\circ\chi)$. Later we will enlarge these algebras by admitting more singular functionals (Section \ref{Peierls} and Appendix \ref{topo}).

In the framework of infinite dimensional differential geometry one models manifolds on general locally convex vector spaces. Many of the results known from the finite dimensional case generalize to this setting. One can define vector fields, differential forms and exterior derivative. We provide precise definitions in the Appendix \ref{iddg}. For further reference see for example \cite{Neeb,Michor}. In particular a locally convex vector space $\E(M)$ is a trivial manifold. General vector fields are derivations of the algebra of smooth functions $\Ci(\E(M))$. Since we restricted our considerations to its subalgebra $\F(M)$, we have to identify a submodule of vector fields that correspond to derivations of $\F(M)$. Moreover we want to associate this module to $M$ in a functorial way. These two requirements already determine the class of vector fields we want to consider.

First we note that vector fields $X$ on $\E(M)$ can be considered as smooth maps from $\E(M)$ to $\E(M)$, since the tangent spaces of a vector space can naturally be identified with the space itself. We restrict ourselves to smooth compactly supported maps $X$ with image in $\D(M)$. Moreover we require them to be local, i.e. $X(\ph)(x)$ depends only on the jet of $\ph$ at the point $x$ (for an equivalent definition of locality, involving the wave front set condition see the Appendix \ref{topo}). 
Vector fields with these properties act on $\F(M)$ as derivations,
\be\label{derivation}
\pa_XF(\ph):=\langle F^{(1)}(\ph),X(\ph)\rangle\,.
\ee
The spacetime support of a vector field $X$ is defined in the following way:
\be\label{suppder}
\begin{split}
\supp\, X=\{x\in M|\forall \text{ neigh. }U\text{ of }x\ & \exists F\in\F(M), \supp\,F\subset U\ \text{ such that }\partial_XF\neq 0 \\
\text{or } \exists\ \ph,\psi\in\E(M),\supp\,\psi\subset U & \text{ such that }X(\ph+\psi)\neq X(\ph)\}\ .
\end{split}
\ee
The space of vector fields with compact support and properties mentioned above can be algebraically completed to an $\F(M)$-module (with respect to the pointwise product (\ref{prod})) which is denoted by $\V(M)$. $\V$ becomes a (covariant) functor by setting
\be\label{funct}
\V\chi(X)= \D\chi\circ X\circ \E\chi \ .
\ee
The action of vector fields on functions and the Lie bracket of vector fields can be considered as special instances of the Schouten bracket of alternating multi-vector fields. In our case these are smooth, compactly supported maps from $\E(M)$ into $\Lambda(\D(M))$, with
\[
\La(\D(M))=\bigoplus \La^n(\D(M))\,,  
\]
where $\La^n(\D(M))$ is the space of compactly supported test functions on $M^n$ which are totally antisymmetric under permutations of arguments (with $\Lambda^0(\D(M))=\RR$).

The alternating multi-vector fields with the regularity properties discussed above form a graded commutative algebra $\La\V(M)$ with respect to the product
\be\label{wedge}
X\wedge Y(\ph)=X(\ph)\wedge Y(\ph)\,,
\ee
where on the right hand side the wedge product is the antisymmetrized tensor product of test functions. 

The Schouten bracket is an odd graded Poisson bracket on this algebra.
It maps
\[
\{\cdot,\cdot\}:\La^n\V(M)\times\La^m\V(M)\to\La^{n+m-1}\V(M)\,,
\]
is graded antisymmetric, i.e.
\[
\{Y,X\}=-(-1)^{(n-1)(m-1)}\{X,Y\}\,,
\]
and satisfies the graded Leibniz rule
\be\label{leibniz}
\{X,Y\wedge Z\}=\{X,Y\}\wedge Z+(-1)^{nm}\{X,Z\}\wedge Y\,,
\ee
where $n$ is the degree of $Y$ and $m$ the degree of $Z$.
For $X\in\La^1\V(M)\equiv\V(M)$ and $F\in\La^0\V(M)\equiv \F(M)$ it coincides with the action of $X$ as a derivation
\[
\{X,F\}=\pa_XF\,,
\]
and for $X,Y\in\La^1\V(M)$ it coincides with the Lie bracket
\[
\pa_{\{X,Y\}}=\pa_X\pa_Y-\pa_Y\pa_X \ .
\]
Moreover, it satisfies the graded Jacobi rule
\be\label{Jacid}
\{X,\{Y,Z\}\}-(-1)^{(n-1)(m-1)}\{Y,\{X,Z\}\}=\{\{X,Y\},Z\} \ , \ n=\mathrm{deg}(X),m=\mathrm{deg}(Y)\,.
\ee

To establish the connection to the BV formalism we identify the functional derivatives $\frac{\delta}{\delta\ph(x)}$ with the antifields $\ph^\ddagger(x)$. The algebra of alternating multivector fields is then the algebra generated by fields and antifields, and the Schouten bracket coincides with the antibracket.
 
In the next step we introduce an action functional $S$. Since neither our spacetimes nor the support of typical configurations are compact we cannot identify $S$ with a function on $\E(M)$. Instead we follow \cite{BDF} and define a Lagrangian $L$ as a natural transformation between the functor of test function spaces $\D$ and the functor $\F_\loc$ such that it satisfies $\supp(L_M(f))\subseteq \supp(f)$ and the additivity rule 
\footnote{We do not require linearity since in quantum field theory the renormalization flow does not preserve the linear structure; it respects, however, the additivity rule (see \cite{BDF})}
\[
L_M(f+g+h)=L_M(f+g)-L_M(g)+L_M(g+h)\,,
\]
for $f,g,h\in\D(M)$ and $\supp\,f\cap\supp\,h=\emptyset$.  
As shown in \cite{BFR} this implies that $L_M(f)$ satisfies an analogous additivity relation (\ref{add}) with respect to the field. Such an additivity relation together with support properties implies that the function $L_M(f)$ is local. 
The action $S(L)$ is now defined as an equivalence class of Lagrangians  \cite{BDF}, where two Lagrangians $L_1,L_2$ are called equivalent $L_1\sim L_2$  if
\be\label{equ}
\supp (L_{1,M}-L_{2,M})(f)\subset\supp\, df\,, 
\ee
for all spacetimes $M$ and all $f\in\D(M)$. 
This equivalence relation applies in particular to Lagrangians differing by a total divergence. 

Following \cite{BDF} we discuss now the equations of motion. The Euler-Lagrange derivative of $S$ is a
natural transformation $S':\E\to\D'$
defined by
\[
\left<S'_M(\ph),h\right>=\left<L_M(f)^{(1)}(\ph),h\right>\,,
 \]
 with $f\equiv 1$ on $\supp h$. The field equation is:
\be
 S_M'(\ph)=0\,.\label{eom}
\ee
The space of solutions $\ph$ is a subspace of $\E(M)$ which we denote by $\E_S(M)$. In the on-shell setting of classical field theory one is interested in the space $\F_S(M)$ of multilocal functionals on $\E_S(M)$. This space can be understood as the quotient $\F_S(M)=\F(M)/\F_0(M)$, where $\F_0(M)$ is the space of multilocal functionals that vanish on $\E_S(M)$ (on-shell). 

The aim of the BV construction is to find a homological interpretation of $\F_S(M)$. In the first step we note that from (\ref{eom}) for every vector field $X\in\V(M)$ the functional
\[
\ph\mapsto\left<S_M'(\ph),X(\ph)\right>=:\delta_S(X)(\ph)\,,
\]
is an element of $\F_0(M)$. Thus we obtain a mapping:
\[
\V(M)\xrightarrow{\delta_S}\F(M)\,,
\]
and $\im\delta_S\subset\F_0(M)$. 
If the equation of motion is  strictly hyperbolic the inclusion holds also in the opposite direction. The proof under some technical assumptions is provided in \cite{BFR}. We say that $\F_0(M)$ is \textit{generated by the equations of motion}. 

Assume that we are given a generalized action functional $S$ for which this is the case. Then we have:
\[
\F_S(M)=\F(M)/\F_0(M)=\F(M)/\im\delta_S\,.
\]
This can be easily translated into the language of homological algebra. Consider a chain complex:
\be\label{Kshort}
\begin{array}{c@{\hspace{0,2cm}}c@{\hspace{0,2cm}}c@{\hspace{0,2cm}}c@{\hspace{0,2cm}}c@{\hspace{0,2cm}}c@{\hspace{0,2cm}}c}
0&\xrightarrow{}&\V(M)&\xrightarrow{\delta_S}&\F(M)&\rightarrow &0\\
 &&1&&0&&
\end{array}
\ee
The numbers below indicate the chain degrees. The $0$-order homology of this complex is equal to: $\F(M)/\F_0(M)=\F_S(M)$. This completes the first step in finding the homological interpretation of $\F_S(M)$. 

In the next step we shall construct a \textit{resolution} of $\F_S(M)$. In homological algebra a resolution of an algebra $A$ is a differential graded algebra $(\A,\delta)$, such that $H_0(\delta)=A$ and  $H_n(\delta)=0$ for $n>0$. We can start constructing the resolution of  $\F_S(M)$ from the chain complex (\ref{Kshort}). We said before that the space of multivector fields  $\La\V(M)$ is a graded commutative algebra with respect to the product (\ref{wedge}). Moreover it is equipped with the natural bracket $\{.,.\}$. 

Since $\delta_S(X)$ is just $\{X,L_M(f)\}$ for $f\equiv 1$ on $\supp X$, for $X\in \V(M)$, we can extend $\delta_S$ to $\La\V(M)$ and obtain the complex:
\be\label{K}
\begin{array}{c@{\hspace{0,2cm}}c@{\hspace{0,2cm}}c@{\hspace{0,2cm}}c@{\hspace{0,2cm}}c@{\hspace{0,2cm}}c@{\hspace{0,2cm}}c@{\hspace{0,2cm}}c@{\hspace{0,2cm}}c}
\ldots&\rightarrow&\La^2\V(M)&\xrightarrow{\delta_S}&\V(M)&\xrightarrow{\delta_S}&\F(M)&\rightarrow& 0\\
 &&2&&1&&0&&
\end{array}\ ,
\ee
where $\delta_S$ is called the Koszul map. Now we want to calculate $H_1(\La\V(M),\delta_S)$.
First we identify the elements of  $\ke(\delta_S)_{\V(M)\rightarrow\F(M)}$. In the BV formalism they are called
\textit{symmetries}. They may be interpreted as the vanishing directional derivatives of the action,
\be\label{sym0}
0=\delta_SX(\ph)=\left<S_M'(\ph),X(\ph)\right>=:\partial_X(S_M)(\ph)\,,
\ee 
for all $\ph\in\E(M)$ (cf. also \cite{Urs}).

A symmetry $X$ is called \textit{trivial} if it vanishes on-shell, i.e. $X(\ph)=0$ for all $\ph\in\E_S(M)$. It is easy to see that $(\im\delta_S)_{\La^2\V(M)\rightarrow\V(M)}$ contains exactly the trivial symmetries, so they don't contribute to $H_1(\La\V(M),\delta_S)$. We can conclude that the first homology is trivial if the action $S$ doesn't possess any \textit{nontrivial local symmetries}. This condition can be formulated as follows:
\begin{equation}\label{sym}
X(\ph)\perp S_M'(\ph)\quad \forall\ph\in\E(M)\Rightarrow X(\ph)=0\ \forall\ph\in\E_S(M)\ .
\end{equation}
Now we derive from (\ref{sym}) a sufficient condition for an action $S$ to be free of nontrivial symmetries.
From $\langle S_M'(\ph),X(\ph)\rangle=0$ it follows in particular that
\[
\frac{d}{d\lambda}\Big|_{\lambda=0}\langle S_M'(\ph+\lambda \psi),X(\ph+\lambda \psi)\rangle=0\quad\forall \psi\in \E(M) \ .
\]
This implies that
\begin{equation} \label{no nontrivial symmetries}
\langle S_M''(\ph),X(\ph)\otimes\psi\rangle+\langle S'(\ph),X^{(1)}(\ph)\psi\rangle=0 \ .
\end{equation}
where the second derivative $S''$ of an action $S\equiv S(L)$ is a natural transformation from $\E\to \D'\otimes \D'$ defined by
\[S''_M= L_M(f)^{(2)} \text{ on }\D(K)\otimes \D(K)\]
with any compact subset $K\subset M$ and with $f\equiv 1$ on $K$. Since $L_M(f)$ is local, its second functional derivative has support on the diagonal, so we can replace one of its arguments by a smooth function without restrictions on the support. Moreover, locality implies that $S_M''(\ph)$ is a differential operator, namely the operator defining the linearized equation of motion around the field configuration 
$\ph$ (in general, $\ph$ may not be a solution of the equation of motion.)
  
Now for $\ph\in \E_S(M)$ the second term in \eqref{no nontrivial symmetries} vanishes, hence we obtain
\[
\langle S_M''(\ph),X(\ph)\otimes\psi\rangle=0\qquad\forall\psi\in\E(M) \ .
\]
This means that for all $\ph\in \E_S(M)$, $X(\ph)$ has to be a solution of the linearized equation of motion. Since $X(\ph)\in\E_c(M)$, the action $S$ possesses no nontrivial symmetries if the linearized equation of motion doesn't have any nontrivial compactly supported solutions. In particular this is the case when $S_M''(\ph)$ is a strictly hyperbolic differential operator.

For an action that doesn't possess any nontrivial symmetries $H_1(\La\V(M),\delta_S)=0$. By an analogous argument one can also prove that the higher homology groups are trivial, thus we have
\begin{eqnarray*}
H_0\left(\La\V(M),\delta_S\right)&=&\F_S(M)\,,\\
H_k\left(\La\V(M),\delta_S\right)&=&0,\ k>0\,.
\end{eqnarray*}
Hence in this case the differential graded algebra $(\La\V(M),\delta_S)$ underlying (\ref{K}) is a resolution of $\F_S(M)$, called the \textit{Koszul resolution}. We want to stress that in our setting $\delta_S$ is not an inner derivation with respect to the antibracket. This is a major difference with respect to other approaches and stems from the fact that the action itself is not an element of $\F(M)$, but has to be understood as an 
equivalence class of natural transformation between the functors $\D$ and $\F$. Nevertheless, locally $\delta_S$ can be written in terms of inner derivations, since 
$\delta_{S(L)}X=\{X,L_M(f)\}$ for $f\equiv 1$ on $\supp\, X$, $X\in \V(M)$.

All the structures presented in this section can be easily generalized to the case of vector-valued fields instead of a scalar field (see \cite{Few}). 
One only has to replace the category $\Loc$ by a category of vector bundles $\Bndl$ over globally hyperbolic spacetimes. If in addition a functor $\B$ from $\Loc$ to $\Bndl$ is given (the case of natural bundles) we obtain again a description in terms of functors on $\Loc$.
\section{The BV formalism for the Yang-Mills theories}\label{YM}
After presenting the basic properties of the Koszul construction on the example of a scalar field we can move to more complicated cases. Assume that we have a configuration space functor $\E$ and a functor $M\mapsto\G(M)$ of Lie groups acting on the configuration spaces. We denote the spaces of orbits under this action by $\E/\G(M)$. Let us take a generalized Lagrangian $L:\D\to\E$ which is invariant under the action of the group and such that the initial value problem is well posed for the elements of $\E/\G(M)$. This is the case for example for Yang-Mills theories. Let $(\E/\G)_S(M)$ denote the space of solutions of the equations of motion. Now we would like to find the space of functionals on $(\E/\G)_S(M)$. We refer to them as gauge invariant on-shell functionals $\F_S^\inv(M)$.
 It would be tempting to repeat exactly the steps of section \ref{scalar} and try to construct the Koszul resolution starting from $\F^\inv(M)$, the space of functions on the gauge orbits $\E/\G(M)$. The problem is that $\E/\G(M)$ in general doesn't admit neither a topological vector space nor a manifold structure. Therefore one cannot apply the geometrical construction of the Koszul resolution directly on $\E/\G(M)$. 
 
Instead we pass over to a cohomological description of  $\F^\inv(M)$ as the 0th order cohomology of the so-called Chevalley-Eilenberg \textit{cochain complex} associated with the induced action of $Lie(\G)$ on $\F(M)$. This is the graded algebra of alternating multilinear maps:
\be\label{CEco}
\CE(M):=\bigoplus\limits_{k=0}^\infty L^k_\alte(Lie(\G),\F(M))\,,
\ee
with a differential $\gamma$, called the Chevalley-Eilenberg map \cite{ChE} which is the external differential with respect to the action of the group. 

Now instead of the Koszul construction based on $\F^\inv(M)=H^0(\CE(M))$ we attempt to construct the Koszul resolution starting from the full graded differential algebra  $\CE(M)$. This amounts to the construction of the Batalin-Vilkoviski complex in the context of gauge theories. Since we want to consider a continuous Chevalley-Eilenberg cohomology, we have to specify a topology on the space of functionals. We choose the nuclear topology $\tau$ of pointwise convergence of functionals together with their derivatives, defined in Appendix \ref{topo}. In this section $\F(M)$ and $\V(M)$ are always equipped with this topology.

After this brief introduction we show the details of the BV construction on the example of Yang-Mills theories. We start this section with some geometrical preliminaries. Then we discuss the functorial aspects of the constructions involved and we introduce the Chevalley-Eilenberg cochain complex. Next we show that this induces a functor from the category of spacetimes to the category of differential graded algebras. Finally we perform the Koszul construction on the differential graded algebra $\CE(M)$ and show that this indeed allows us to recover the space of gauge invariant on-shell functionals $\F_S^\inv(M)$ as a certain cohomology class. This construction is usually called Koszul-Tate resolution. We end this section with the discussion of the so called \textit{classical master equation}. In the standard approach this is a condition which has to be satisfied by the action functional extended to the BV-complex. In our framework this is a condition formulated on the level of natural transformations.
\subsection{Geometrical preliminaries}\label{geom}
 Let $G$ be a semisimple compact Lie group and $g$ its Lie algebra.
Let $P\xrightarrow{\pi} M$ be a principal bundle with the fiber $G$. We restrict ourselves to trivial bundles in order to have a functor from $\Loc$ to the category of principal bundles (see the remark at the end of Section \ref{scalar}). Then the configuration space of the theory is $\E(M)=\Omega^1(P,g)^G\cong \Omega^1(M,g)$. 
Now let $\chi:M\rightarrow N$ be a causal isometric embedding. We can define a morphism from $\E(N)$ to $\E(M)$ in a natural way by setting: $\E\chi(\omega\otimes a):=\chi^*\omega\otimes a$, where $\omega\in\Omega^1(M)$, $a\in g$ and the pullback of a differential form is defined as:
 $(\chi^*\omega)_x:=\omega_{\chi(x)}\circ d_x\chi$. In this way $\E$ becomes a contravariant functor between the categories $\Loc$ and $\Vect$. One can also define a covariant functor $\E_c$ by assigning to a spacetime the space of compactly supported $g$-valued forms $\Omega_c^1(M,g)$. In this case ${\E_c}\chi$ maps forms to their push-forwards.
 
We define the gauge group as the space of vertical $G$-equivariant compactly supported diffeomorphisms of $P$:
 \[
 \G:=\{\alpha\in \Diff_c(P)|\alpha(p\cdot g)=\alpha(p)\cdot g, \pi(\alpha(p))=\pi(p),\ \forall g\in G, p\in P\}\,.
 \]
This space can be also characterized by $\G\cong\Gamma_c(M\leftarrow(P\times_G G))$. For a trivial bundle $P$ this is just $\G(M)\cong\Ci_c(M,G)$. It was shown (\cite{Neeb04,Gloe,Michor}, see also \cite{Neeb,Wock}) that $\Ci_c(M,G)$ can be equipped with a structure of an infinite dimensional Lie group modeled on its Lie algebra $\frakg(M):=\Ci_c(M,g)$. The exponential mapping can be defined and it induces a local diffeomorphism at 0. Since the gauge group is just a subgroup of $\Diff(P)$, it has a natural action on $\Omega^1(P,g)^G$ by the pullback:
\[\rho_M(\alpha)A=(\alpha^{-1})^*A, \quad\alpha\in\mathcal{G},\quad A\in\Omega^1(P,g)^G\] 
The derived action of the Lie algebra $\frakg$ on $\Omega^1(P,g)^G$ is therefore defined as:
\be\label{rho}
\rho_M(\xi)A\doteq\frac{d}{dt}\Big|_{t=0}\rho_M(\exp t\xi)A=\frac{d}{dt}\Big|_{t=0}(\exp(-t\xi))^*A=\pounds_{Z_{\xi}}A=d\xi+[A,\xi]\,,
\ee
where $Z_{\xi}$ is the fundamental vector field on $P$ associated to $\xi$. 
$\rho(\xi)$ can be interpreted as a vector field on $\E(M)$ (in analogy to the definition of vector fields $X\in \V(M)$ in Section \ref{scalar}) which associates to the field configuration $A\in\E(M)$ the compactly supported configuration $d\xi+[A,\xi]$. Clearly, $\rho$ is a natural transformation from $\frakg$ to $\V$ as may be seen from the relation
\[
\rho_N(\chi_*\xi)A=\chi_*(\rho_M(\xi)(\chi^*A))
\]
for a causal embedding $\chi:M\to N$.

Now we introduce the generalized Lagrangian $L_M(f)=-\frac{1}{2}\int_M f\,\tr(F \wedge * F)$, where $F=dA+[A,A]$ is the field strength corresponding to the gauge potential $A$ and $*$ is the Hodge operator. Using the criterion (\ref{sym}) we see that this action has nontrivial symmetries because the linearized equation of motion might possess nontrivial compactly supported solutions. Actually these symmetries can be easily characterized. We see immediately that for each $\xi\in\frakg(M)$ we obtain a symmetry $\rho_M(\xi)\in\V(M)$ by the invariance of the Lagrangian. 
More general symmetries may be obtained by extending $\rho_M$ to an $\F(M)$-module map $\fG(M)\to\V(M)$ where $\fG(M)$ is the space of smooth compactly supported multilocal functions $\Xi:\E(M)\rightarrow\frakg(M)$ such that $\bigcup\limits_{\ph\in\E(M)}\supp(\Xi(\ph))$ is compact. The last condition is needed to assure that $\rho_M(\Xi)$ is indeed a compactly supported vector field according to the definition (\ref{suppder}). Any symmetry may be obtained by a sum of a trivial symmetry with a symmetry of the form $\rho_M(\Xi)$ with $\Xi\in\fG(M)$.
\subsection{Chevalley-Eilenberg complex} \label{ChEil}
We start this section with the construction of the Chevalley-Eilenberg cochain complex $\CE(M)$. In order to have a good behavior of this structure under the embeddings $\chi:M\rightarrow N$, we drop the assumption of compact support of the elements of $\frakgo(M):=\Ci(M,g)$ and require instead  the linear alternating maps in (\ref{CEco}) to be compactly supported. The assignment of $\frakgo(M)$ to a manifold $M$ is a contravariant functor. It associates to a morphism $\chi:M\rightarrow N$ a map $\frakgo\chi$ acting on functions as a pullback: $\frakgo\chi(f\otimes a):=\chi^*f\otimes a$ for $f\in\Ci(N)$, $a\in g$. 

Now instead of adapting (\ref{CEco}), we define $\CE(M)$ to be the space of smooth compactly supported multilocal maps $\Ci_\ml(\E(M),\La{\frakgo}'(M))$ (the definition of multilocal vector-valued maps is given in the appendix \ref{topo}). Here ${\frakgo}'(M)$ means the strong dual of $\frakgo(M)$ (in other words ${\frakgo}'(M)$ is the space of compactly supported distributions with values of $g$) We equip $\CE(M)$ with the topology $\tau$ of pointwise convergence of functionals together with their derivatives (see Appendix \ref{topo}). Note that $\CE(M)$ has $\La^k{\frakgo}'(M)\otimes\F(M)$ as a dense subspace\footnote{Since the spaces $\frakgo$ and $\frakg$ are nuclear, the projective and injective tensor products coincide (see Appendix \ref{topo} and \cite{Jar} for details).}, so our definition can be seen as a generalization of the standard one used in finite dimensional case. The Chevalley-Eilenberg  differential $\gamma_M$ is defined on the tensor products as:
\begin{eqnarray}
\gamma_M:&\ &\Lambda^q\frakgo'(M)\otimes\F(M)\rightarrow\Lambda^{q+1}\frakgo'(M)\otimes\F(M)\,,\nonumber\\
(\gamma_M \omega)(\xi_0,\ldots, \xi_q)&\doteq&\sum\limits_{i=0}^q(-1)^i\partial_{\rho_M(\xi_i)}(\omega(\xi_0,\ldots,\hat{\xi}_i,\ldots,\xi_q))+\nonumber\\
&+&\sum\limits_{i<j}(-1)^{i+j}\omega\left([\xi_i,\xi_j],\ldots,\hat{\xi}_i,\ldots,\hat{\xi}_j,\ldots,\xi_q\right)\,,
\end{eqnarray}
and extended to the whole space $\CE(M)$ by continuity.
Let $\dgA$ be the category with differential graded algebras as objects and differential graded algebra embeddings as morphisms. From the naturality of $\rho$ and $[.,.]$ it follows that $\CE(M)$ becomes a covariant functor from $\Loc$ to $\dgA$ if we set $\CE\chi(\omega):=(\frakgo'\chi)^k\circ\,\omega\circ\E\chi$, for $\omega\in \Ci(\E(M),\Lambda^{k}\frakgo'(M))$. The space of gauge invariant functionals is now recovered as: $\F^\inv(M)=H^0(\CE(M),\gamma_M)$. This is again a covariant functor.

To end this section we want to make a brief comment on the relation to the "standard" approach. 
Note that we can write elements of 
$\CE$ formally as:
\[
F(\ph)=\int dx^1\ldots dx^n f_{a_1...a_n}(\ph)(x_1,\ldots, x_n) C(x_1)^{a_1}\wedge...\wedge C(x_n)^{a_n}\,,
\]
where $f_{a_1...a_n}(\ph)$ is a compactly supported distribution and $C^a(x)\in  {\frakgo}'$ are coefficients of the Maurer-Cartan form $C$ on $\G(M)$. In physics one calls them the \textit{ghost fields}. They can be seen as generators of the algebra $\CE(M)$. In the present setting they appear naturally as elements of the Chevalley-Eilenberg complex.
The grading of this complex is called the \textit{pure ghost number}. We denote it by $\#\pg$.
\subsection{Koszul-Tate resolution}\label{KT}
In the previous section we showed that $\F^\inv(M)=H^0(\CE(M),\gamma_M)$. Now we want to construct  $\F_S^\inv(M)$, the space of gauge invariant on-shell functionals. Obviously we have:
 $\F_S^\inv(M)=H^0(\CE_S(M),\gamma_M)$, where $\CE_S(M)$ is the space of smooth compactly supported multilocal on-shell functions $\Ci_\ml(\E_S(M), \La{\frakgo}'(M))$. Therefore our task now is to find the Koszul resolution of the differential graded algebra $\CE_S(M)$. We shall do it by repeating the steps of Section \ref{scalar}, but this time for a graded configuration space. Vector fields are replaced with graded derivations of the algebra $\CE(M)$. We restrict ourselves to those which are smooth compactly supported (as derivations, i.e. in the sense of (\ref{suppder})) multilocal maps
 \be\label{DerG}
 \E(M)\rightarrow\La{\frakgo}'(M)\widehat{\otimes}\left(\E_c(M)\oplus\frakg(M)\right)\,.
\ee
A pair $(X,\xi): \E(M)\rightarrow \E_c(M)\oplus\frakg(M)$ acts on $\CE(M)$ in the following way:
\be \label{interior}
(\partial_{(X,\xi)}F)(\ph):=(\partial_XF)(\ph)+i_{\xi(\ph)} F(\ph)\,,
\ee
where $i_{\xi(\ph)}$ is the interior product, i.e. the insertion of $\xi(\ph)\in\frakg(M)$ into
$\La\frakgo'$. The first term on the right hand side of \eqref{interior} is an even and the second one an odd derivation. The action of a general derivation of the form (\ref{DerG}) can be now defined by imposing the graded distributive rule. The space of such derivations is denoted by $\Der(\CE(M))$. Note that it contains in particular $\V(M)$ and $\fG(M)$. We extend the grading $\#\pg$ of $\CE(M)$ to a grading $\#\gh$ (called total ghost number) on $\Der(\CE(M))\oplus\CE(M)$ by 
\[
\#\gh=\#\pg-\#\af
\]
The antifield number $\#\af=1$ is assigned to the vector fields ($\E_c(M)$-valued functions), the antifield number $\#\af=2$ to the elements of $\fG(M)$ ($\frakg(M)$-valued functions), whereas elements of $\CE(M)$ have $\#\af=0$.
The graded commutator $[.,.]$ on $\Der(\CE(M))$ and the evaluation of a derivation on an element of $\CE(M)$ are special instances of the Schouten bracket $\{.,.\}$ on   $\Der(\CE(M))\oplus\CE(M)$, equipped with the grading $\#\gh$. Like in the scalar case, this structure is called \textit{the antibracket}. 

Using the fact that we restricted ourselves to derivations with compact support, one can show that $\Der(\CE)$ is a covariant functor from $\Loc$ to $\Vect$. The Chevalley-Eilenberg differential $\gamma_M$ itself is not an element of $\Der(\CE(M))$, since it is not compactly supported. Compare this with a similar situation encountered in Section \ref{scalar} when we showed that $\delta_S$ is not an inner derivation with respect to the antibracket. The differential $\gamma_M$ can be decomposed as a sum of two derivations $\gamma_M=\gamma_M^{(0)}+\gamma_M^{(1)}$, where:
\begin{eqnarray}
(\gamma_M^{(0)}F)(\xi)&\doteq&\partial_{\rho_M(\xi)}F,\quad\quad F\in\F(M),\ \xi\in\frakgo(M)\,,\label{gamm}\\
(\gamma_M^{(1)}\omega)(\xi_1,\xi_2)&\doteq&-\omega([\xi_1,\xi_2]),\quad\ \omega\in \frakgo'(M),\ \xi_1,\xi_2\in\frakgo(M)\nonumber
\end{eqnarray}
and $\gamma_M^{(0)}$, $\gamma_M^{(1)}$ are extended to the whole space $\CE(M)$ by requiring the graded Leibniz rule. Although $\gamma_M$ is not inner with respect to $\{.,.\}$ we can consider a following family of mappings $\theta_M$ from $\D(M)$ to $\Der(\CE)(M)$
\begin{eqnarray*}
(\theta_M^{(0)}(f)F)(\xi)&\doteq&\partial_{f\rho_M(\xi)}F,\quad\quad F\in\F(M),\ \xi\in\frakgo(M)\,,\\
(\theta_M^{(1)}(f)\omega)(\xi_1,\xi_2)&\doteq&\omega(f[\xi_1,\xi_2]),\quad\ \omega\in\frakgo'(M),\ \xi_1,\xi_2\in\frakgo(M)\,,
\end{eqnarray*}
where $f\in\D(M)$ is a test function and the multiplication with elements of $\E_c(M)$ and $\frakgo(M)$ is defined pointwise. It follows now that $\{\omega,\theta_M(f)\}=\gamma_M(\omega)$ if $\supp(\omega)\subset f^{-1}(1)$, $\omega\in\CE(M)$. From the fact that the Lie-algebra action is local it follows that $\theta$ constructed in this way is a natural transformation between the functors $\D$ and $\Der(\CE)$. One sees immediately the analogy with the generalized Lagrangians. We can now introduce a differential $s$ on the $\CE(M)$-module $(\Der(\CE)\oplus\CE)(M)$ by the following definition:
\be\label{s}
sF=\{F,L_M(f)+\theta_M(f)\}\,,
\ee
where $f\equiv1$ on $\supp\, F$, $F\in (\Der(\CE)\oplus\CE)(M)$. The graded differential module  $((\Der(\CE)\oplus\CE)(M),s)$ can be extended to a  graded differential algebra by adding symmetric powers of even elements and antisymmetric powers of odd elements. The resulting structure is called the Batalin-Vilkovisky complex $\BV(M)$. It is the space of smooth compactly supported multilocal mappings
\be\label{BV}
\BV(M)=\Ci_\mc(\E(M),\A(M))\,,
\ee
where we denoted $\A(M):=\Lambda\E_c(M)\widehat{\otimes}S^\bullet \frakg(M)\widehat{\otimes}\Lambda {\frakgo}'(M)$. 
This graded algebra can be explicitly characterized as a space of distributional sections of a certain vector bundle over $M$ (see the Appendix \ref{topo} for details).

The map $s$ defined in (\ref{s}) can be extended to a graded differential on $\BV(M)$ by means of the graded Leibniz rule. 
The grading of $\BV(M)$ by the antifield number $\#\af$ is compatible with the product, but not with the antibracket. Also the differential $s$ is not homogeneous with respect to this grading.
We can now expand the differential $s$ with respect to the antifield number. The expansion has two terms: $s=s^{(-1)}+s^{(0)}$, which can be identified as follows:
\begin{itemize}
\item $s^{(-1)}$ is the Koszul-Tate differential providing the resolution of $\CE_S(M)$,
\item $s^{(0)}$ is the Chevalley-Eilenberg differential on   $\CE_S(M)$.
\end{itemize} 
This results in the following bicomplex structure\footnote{We omit the dependency on $M$, since all the maps are natural and can be written on the level of functors.}:\\
\[
\begin{CD}
\ldots@>s^{(-1)}>>\big(\La^2\V\oplus\fG\big) @>s^{(-1)}>>\V@>s^{(-1)}>>\F@>s^{(-1)}>>0\\ 
@.     @VV{s^{(0)}}V@VV{s^{(0)}}V@VV{s^{(0)}}V@.\\
\ldots@>s^{(-1)}>>{\Ci_\ml\big(\E,(\La^2\E_c\oplus\frakg)\widehat{\otimes}\frakgo'}\big)@>s^{(-1)}>>{\Ci_\ml\big(\E,\E_c\widehat{\otimes}\frakgo'\big)}@>s^{(-1)}>>{\Ci_\ml\big(\E,\frakgo'\big)}@>s^{(-1)}>>0
\end{CD}
\]
The first row of this bicomplex corresponds to the resolution of $\F_S(M)$. This can be easily seen, since  $s^{(-1)}$ on $\fG(M)$ is just $\rho_M$ and $\im(\rho_M)_{\fG(M)\rightarrow\V(M)}=\ke(\delta_S)_{\V(M)\rightarrow\F(M)}$. Moreover one can prove that $\F_0(M)$ in case of Yang-Mills theories is also generated by the equations of motion\footnote{For local functionals this was shown within the jet bundle formalism for example in \cite{HennBar}. This result can be generalized to more singular functionals, like the microcausal ones, by taking the sequential completion. See appendix \ref{topo} for details on the topology.}, i.e. $\F_0(M)=\im(\delta_S)_{\V(M)\rightarrow\F(M)}$.
Explicitly the first row of the bicomplex can be written as:
 \[
\ldots\rightarrow\La^2\V\oplus\fG\xrightarrow{\delta_S\oplus\rho}\V\xrightarrow{\delta_S}\F\rightarrow 0
\]
One can check that the 0-order homology of this complex is just $\F_S(M)$ and the higher homology groups are trivial.

Since $(\BV(M),s^{(-1)})$ is a resolution we can use a standard result in homological algebra \footnote{See \cite{Henneaux:1992ig} and references therein for the review and \cite{Richter,Weibel,Eisen} for the mathematical details.} and write the cohomology groups of $s$ in the form:
\[
H^k(\BV(M),s)=H^k(H_0(\BV(M),s^{(-1)}),s^{(0)})
\]
Because $H_0(\BV(M),s^{(-1)})=\CE_S(M)$ we have in particular:
\[
H^0(\BV(M),s)=H^0(\CE_S(M),s^{(0)})=\F^\inv_S(M)
\]
This shows that we can recover the information on gauge invariant on-shell functionals from the graded differential algebra $(\BV(M),s)$. Moreover, since all the steps were done in the covariant way, we conclude that $\BV$ can be made into a functor from $\Loc$ to $\dgA$.
\subsection{Classical master equation}\label{CME}
Finally we want to comment on one more aspect of the BV-construction, namely the \textit{classical master equation}. In our setting this has to be understood on the level of natural transformations. First recall that the generalized Lagrangians are natural transformations between the functors $\D$ and $\F_\loc$. In the BV construction we extended the space of functionals $\F(M)$ to the BV complex $\BV(M)$. Let $\BV_\loc(M)$ denote the linear subspace of  $\BV(M)$ consisting of functionals that are local with respect to all the variables\footnote{The definition of locality for the elements of $\CE(M)$ and $\BV(M)$ is given in the Appendix \ref{topo}.}.
We can generalize the notion of a Lagrangian to a natural transformation between the functors $\D$ and $\BV_\loc$. Let $\Nat( \D,\BV_\loc )$ denote the set of natural transformations\footnote{It was shown in \cite{Few} that this is indeed a small set.}. Since we want to include also products of local functions in our discussion, a structure more general than $\Nat( \D,\BV_\loc )$ is needed. Let $\D^k$ be a functor from the category $\Loc$ to the product category ${\Vect}^k$, that assigns to a manifold $M$ a $k$-fold product of the test section spaces $\D(M)\times\ldots\times \D(M)$. Let $\Nat(\D^k,\BV_\loc)$ denote the set of natural transformations from $\D^k$ to $\BV_\loc$. We define extended Lagrangians $L\in Lgr$ to be elements of the space $\bigoplus_{k=0}^\infty \Nat(\D^k,\BV_\loc)$ satisfying: $\supp(L_M(f_1,...,f_n))\subseteq \supp f_1\cup...\cup\supp f_n$ and the additivity rule in each argument. We can introduce on $Lgr$ an equivalence relation similar to (\ref{equ}). We say that $L_1\sim L_2$, $L_1,L_2\in\Nat(\D^k,\BV_\loc)$ if:
\be\label{equ2}
\supp((L_1-L_2)_M(f_1,...,f_k))\subset \supp(df_1)\cup...\cup\supp(df_k),\qquad\forall f_1,...,f_k\in\D^k(M)
\ee
The natural transformation $L^{\ex}:=L+\theta$ is an example of a generalized Lagrangian in $Lgr$. We call the corresponding equivalence class $S^{\ex}$ the \textit{extended action}. As noted before $sF=\{F,L_M^\ex(f)\}$, for $f\equiv 1$ on the support of $F$, $F\in \BV(M)$.  To make a contact with the standard approach we shall write  $L^{\ex}$ explicitly:
\[
L^{\ex}_M(f)=-\frac{1}{2}\int_M f\tr(F\wedge *F)+\int_M f\big(dC+\frac{1}{2}[A,C]\big)^I_\mu(x)\frac{\delta}{\delta A^I_\mu(x)}+\frac{1}{2}\int_M f [C,C]^I(x)\frac{\delta}{\delta C^I(x)}\,.
\]
This is the standard extension of the Yang-Mills action in the BV-formalism. The antibracket can be lifted to a bracket on $Lgr$ by the definition:
\be\label{ntbracket}
\{L_1,L_2\}_M(f_1,...,f_{p+q})=\frac{1}{p!q!}\sum\limits_{\pi\in P_{p+q}}\{{L_1}_M(f_{\pi(1)},...,f_{\pi(p)}),{L_2}_M(f_{\pi(p+1)},...,f_{\pi(p+q)})\}\,,
\ee
where $P_{p+q}$ denotes the permutation group.
 The lifted antibracket is again graded antisymmetric and satisfies the graded Jacobi identity (\ref{Jacid}). It will be shown in section \ref{ntBV} that if we extend $Lgr$ to a graded algebra, (\ref{ntbracket}) satisfies also the Leibniz rule and is therefore a graded Poisson bracket.
 
The classical master equation extended to the natural transformations (ECME) can now be formulated as:
\be
\{L^{\ex},L^{\ex}\}\sim 0\,,
\ee 
with the equivalence relation defined in (\ref{equ2}). It guarantees the nilpotency of $s$ defined by $sF=\{F,L^{\ex}_M(f)\}$, where $f\equiv1$ on $\supp F$, $F\in \BV(M)$.
\section{Gauge fixing}\label{gaugefixing}
From the physical point of view, the BV-complex is an important structure since it makes it possible to perform gauge-fixing in the systematic way in the Lagrangian formalism. This point is usually not addressed in the mathematical literature, since the gauge fixing procedure involves some amount of arbitrariness. In physics one fixes the gauge to introduce the dynamical Poisson structure on the algebra of gauge invariant observables. This is the so called Peierls bracket \cite{Pei,Mar}. In the BRST-formalism one introduces a Peierls bracket first on the extended algebra and then shows that it is also well defined on the cohomology classes. This can be done systematically with the help of the Batalin Vilkovisky complex. 
 In the BV framework gauge fixing means eliminating the antifields by setting them equal to some functions of fields \cite{Batalin:1981jr,Henneaux:1989jq,Henneaux:1992ig,Froe}.
 
The gauge fixing is usually done in two steps. First one performs a transformation of  $\BV(M)$, that leaves the antibracket $\{.,.\}$ invariant. This provides us with the new extended action $\tilde{S}^{\ex}$ and a new differential $\tilde{s}$. Since the transformation leaves the antibracket invariant, we have an isomorphism of the cohomology classes $H^0(\BV(M),s)\cong H^0(\BV(M),\tilde{s})$. In the second step we want to set the antifields to 0. This can be done in a systematic way by introducing a new grading on $(\BV(M)$, the so called \textit{total antifield number} $\#\ta$. It has value $0$ on fields and value $1$ on all antifields. Next we expand the differential $\tilde{s}$ with respect to this new grading: $\tilde{s}=\delta^g+\gamma^g+\ldots$ (for Yang-Mills theories this expansion has only two terms). From the nilpotency of $\tilde{s}$ it follows that $\delta^g$ is a differential and $\gamma^g$ is a differential modulo $\delta^g$. Moreover $\delta^g$ can be interpreted as the Koszul map corresponding to the so called "gauge-fixed action" $S^g$. We have to choose the canonical transformation of $\BV(M)$ in such a way that this extended action doesn't have nontrivial symmetries. In this case the Koszul map  $\delta^g$ provides a resolution and using the main theorem of homological perturbation theory \cite{Henneaux:1992ig,Barnich:1999cy} one can conclude that:
\be\label{gfix}
H^0(\BV(M),\tilde{s})=H^0(H_0(\BV(M),\delta^g),\gamma^g)\,.
\ee
The r.h.s of (\ref{gfix}) is called the \textit{gauge-fixed cohomology}. It was discussed in details in \cite{Barnich:1999cy}. It was argued in \cite{Henneaux:1992ig,Barnich:1999cy} that we can view $H_0(\BV(M),\delta^g)$ as the Koszul-Tate resolution for the action $\tilde{S}^{\ex}$ where antifields are set to 0. 
\subsection{Nonminimal sector}
We describe now the above procedure in details on the example of Yang-Mills theories. There is one more remark that has to be made at this point. To implement the usual gauge fixing (for example the Lorenz one) we need first to introduce Lagrange multipliers. In the spirit of classical Lagrangian field theory, these are auxiliary, non-physical fields, that have to be eliminated at the end by performing a quotient of the field algebra. In the homological framework we have to introduce them in a way, that would not change the cohomology classes of $s$. The natural way to do it is to extend $\BV(M)$ by so called contractible pairs. 

Two elements $a$, $b$ of the cochain complex with a differential $d$ form a contractible pair if $a=db$ and $a\neq0$. Let $a$ be of degree $n$ and $b$ of degree $n-1$. Since $H^n(d)=\ke(d_n)/\im(d_{n-1})$, $a$ and $b$ are mapped to the trivial elements  $[0]$ of the cohomology classes. This observation provides us with a method to add Lagrange multipliers to the BV-complex. For concreteness we take the Lorenz gauge $G(A)=\hinv d*\!A$, where $*$ denotes the Hodge dual and $G$ is a map from $\E(M)$ to $\Ci(M,g)$. This suggests that the Lagrange multipliers (also called Nakanishi-Lautrup fields) should be elements of $\Ci(M,g')$ which we can also identify with $\Ci(M,g)$ because of the duality on $g$. Therefore we extend the BV-complex by tensoring with the space: $S^\bullet\frakgo'(M)$, which is interpreted as functionals of the  Lagrange multipliers and have grade $\#\gh=0$. Together with this space we introduce the space $\La\frakgo'(M)=\bigoplus\limits_{k=0}^\infty\La^k\frakgo'(M)$. These are the functionals of the so called \textit{antighosts} and have $\#\gh=-k$. They form trivial pairs with Nakanishi-Lautrup fields if we define: $sF=0$, and $sG=\Pi G\circ m_i$ for $F\in S^1\frakgo'(M)$, $G\in\La^1\frakgo'(M)$, where $\Pi$ denotes the grade shift by $+1$ and $m_i$ the multiplication of the argument by $i$. The last operation is just a convention used in physics to make antighosts hermitian. We use it to stay consistent with the literature. Together with antighosts and Nakanishi-Lautrup fields we can introduce the corresponding antifields (derivations of  $S^\bullet\frakgo'(M)$ and $\La\frakgo'(M)$). The full nonminimal sector is of the form:
\[
\Nm(M)=\La \frakgo' [-1]\widehat{\otimes}S^\bullet\frakgo'(M)[0]\widehat{\otimes}S^\bullet{\frakg}(M)[0]\widehat{\otimes}\La {\frakg}[-1]\,,
\]
where we indicated the grades explicitly in brackets. The new BV-complex $\BVn(M)$ consists of compactly supported multilocal maps $\Ci_\ml(\E(M),\Nm(M)\widehat{\otimes}\A(M))$ with the BV-differential $s$ defined above. It can be seen that $sF=\{F,L^\ex(f)\}$, where $f\equiv 1$ on $\supp F$, $F\in \BVn(M)$ and the extended Lagrangian is now:
\begin{multline}\label{Snm}
L^{\ex}_M(f)=-\frac{1}{2}\int_M f\tr(F\wedge *F)+\int_M f\big(dC+\frac{1}{2}[A,C]\big)^I_\mu(x)\frac{\delta}{\delta A^I_\mu(x)}+\\
+\frac{1}{2}\int_M f [C,C]^I(x)\frac{\delta}{\delta C^I(x)}-i\int_M f B_I(x)\frac{\delta}{\delta \bar{C}_I(x)}\,.
\end{multline}
The last term corresponds to the action of $s$ on the non-minimal sector and we used the traditional notation $B$ for Nakanishi-Lautrup fields and $\bar{C}$ for the antighosts\footnote{Expression (\ref{Snm}) is a little bit formal. We can make it precise if we treat $B^I(x)$ as an evaluation functional on the space of Lagrange multipliers ${\frakgo}(M)$, i.e. $B^I(x)\in S^1\frakgo'(M)$. In this sense elements of $S^\bullet\frakgo'(M)$ can be seen as integrals $F=\int dx_1...dx_n f_{a_1...a_n}(x_1,...,x_n) B^{a_1}(x_1)\otimes...\otimes B^{a_n}(x_n)$. Similarly we can write $G\in \La\frakgo'(M)$ as $G=\int dx_1...dx_n g_{a_1...a_n}(x_1,...,x_n) \bar{C}^{a_1}(x_1)\wedge...\wedge  \bar{C}^{a_n}(x_n)$, for evaluation functional $\bar{C}^{a_n}(x_n)$. Then the action of $s$ on the non-minimal  sector can be written as: $sB^I(x)=0$, $s\bar{C}^I(x)=iB^I(x)$. The factor $i$ is only a convention.
 If we identify $\frac{\delta}{\delta \bar{C}^I(x)}$ with the derivation that acts on $\La\frakg(M)$ as the left derivative (see \cite{Rej} for the detailed discussion), we arrive at the expression (\ref{Snm}).}.
\subsection{Gauge fixing for the Yang-Mills theory}
Now we turn back to the gauge-fixing. 
Let $\psi\in\BVn(M)$ be a fixed algebra element of degree $\#\gh=-1$ and  $\#\af=0$. 
Using $\psi$ we define a linear transformation $\alpha_\psi$ on $\BVn(M)$ 
by
\be\label{gfermion}
\alpha_\psi(X):=\sum_{n=0}^{\infty}\frac{1}{n!}\underbrace{\{\psi,\dots,\{\psi}_n,X\}\dots\}\,,
\ee
The antibracket with $\psi$ preserves the ghost number $\#\gh$ and lowers the antifield number $\#\af$ by 1. Hence the sum in \eqref{gfermion} is finite and $\alpha_\psi$ preserves the grading with respect to the ghost number. Moreover, it preserves as well the product (as a consequence of the Leibniz rule for the Schouten bracket) as well as the Schouten bracket itself (as a consequence of the Jacobi identity).  

Let now $\Psi$ be a natural transformation from $\D$ to $\BVn$ such that $\Psi_M(f)$ satisfies the conditions which were stated on $\psi$ above. $\Psi$ is called \textit{the gauge fixing fermion}.
We define an automorphism on $\BVn(M)$ by
\[
\alpha_{\Psi_M}(X)=\alpha_{\Psi_M(f)}(X),\qquad X\in\BVn(M)\,
\]
where $f\equiv1$ on the support of $X$.

Let $\tilde{L}^{\ex}=\alpha_\Psi\circ L^\ex$ be the transformed generalized Lagrangian. 
We define a new BV-operator as $\tilde{s}F:=\{F,\tilde{L}^\ex_M(f)\}$, for $f\equiv 1$ on $\supp F$, $F\in\BVn(M)$ . We have:
\[
\F_S^{\inv}(M)\cong H^0(\BVn(M),\tilde{s})\,,
\]
where the isomorphism is given by means of $\alpha_\Psi$. 

For the Lorenz gauge we choose the gauge-fixing fermion of the form:
\be\label{Lorenz}
\Psi_M(f)=i\int\limits_M f\left(\frac{\alpha}{2}\kappa(\bar{C},B)+\left<\bar{C},*d*\!A\right>_g\right)\dvol\,,
\ee
where $\kappa$ is the Killing form on the Lie algebra $g$ and $\left<.,.\right>_g$ is the dual pairing between $g$ and $g'$.

In the second step of the gauge fixing procedure we expand the differential $\tilde{s}$ with respect to the total antifield number: $\tilde{s}=\delta^g+\gamma^g$, where $\delta^g$ lowers $\#\ta$ by 1 and $\gamma^g$ preserves it. Let $X\in\BVn$ be a derivation of total antifield number $\#\ta=1$. The action of $\delta^g$ on $X$ is given by:
\be
\delta^gX=\{X,\tilde{L}^\ex_M(f)\}\Big|_{\#\ta=0\atop \textrm{terms}}=\{X,L^g_M(f)\},\qquad f\equiv 1\textrm{ on }\supp X\,,\label{deltag}
\ee
where $L^g$ is the so called gauge-fixed Lagrangian and is obtained from $\tilde{L}$ by putting all antifields to 0. The corresponding equivalence class $S^g$ is the gauge-fixed action. The ideal of $\BVn(M)$ generated by all terms of the form (\ref{deltag}) is the graded counterpart of the ideal of $\F(M)$ generated by the equations of motion. In the next section we shall see that one can introduce a notion of a derivative on $\BVn(M)$ which makes this correspondence precise. In this sense the $0$-order homology of $\delta^g$ is the algebra of on-shell functions for the gauge-fixed action $S^g$. For Yang-Mills theory the gauge-fixed Lagrangian reads:
 \[
L_M^g(f)=S_M(f)+\gamma^g\Psi_M(f)\,.
\]
In case of the Lorenz gauge we obtain:
 \be
L_M^g(f)=-\frac{1}{2}\int\limits_M f\tr(F\wedge *F)-i\int\limits_M f\tr[d\bar{C},*D C]-\int\limits_M f\left(\frac{\alpha}{2}\kappa(B,B)+\left<B,\hinv d*\!A\right>_g\right)\dvol\,.\label{fixed}
\ee

The differential $\gamma^g$ is called \textit{the gauge-fixed BRST differential}. 
The action of the gauge-fixed BRST differential on the functions in $\BVn(M)$ is summarized in the table below.\\
\begin{center}
{\setlength{\extrarowheight}{2.5pt}
\begin{tabular}{ll}
\toprule%
& $\gamma^g$\\\otoprule%
$F\in\F(M)$&$\left<F^{(1)},dC+[.,C]\right>$\\
 $C$&$-\frac{1}{2}[C,C]$\\
 $B$& $0$\\
  $\bar{C}$& $iB$\\\bottomrule
\end{tabular}}
\end{center}
\section{Peierls bracket}\label{Peierls}
We come finally to the discussion of the dynamical structure. As already mentioned in Section \ref{gaugefixing}, we first need to fix the gauge, before we can define the Peierls bracket, that implements the dynamics. As discussed in  Appendix \ref{topo} the space of multilocal functionals is not closed under the Peierls bracket. To fix this we replace it by the space of the space $\F_\mc(M)$ of \textit{microcausal functionals}, equipped with the topology $\tau_\Xi$, described in Appendix \ref{topo}.  Multilocal functionals are dense in $\F_\mc(M)$ with respect to this topology. Using this space as a starting point one can repeat the construction of the BV-complex given in section \ref{KT} with some technical changes discussed in Appendix \ref{topo}. The elements of $\BV(M)$ are now microcausal vector-valued functions\footnote{The notion of microcausal $\CC$-valued functionals was introduced in \cite{BFK95}, see also  \cite{BDF,BFR}. For the definition of microcausal elements of the BV-complex see Appendix \ref{topo}.}: $\BV(M)=\Ci_\mc(\E(M),\A)$. For Yang-Mills theory in the BV complex extended by the nonminimal sector $\A$ is of the form:
\[
\A=\prod\limits_{k,l,m=0}^\infty\Gamma'_{\Xi_{n}}(M^n,S^kg\otimes\Lambda^{l}g\otimes\Lambda^{m}g\otimes\textrm{Antifields}),\quad k+l+m=n\,.
\]
In the above formula $\Xi_n$ denotes the open cone $\{(x_1,...,x_n,k_1,...k_n)| (k_1,...k_n)\notin (\overline{V}_+^n \cup \overline{V}^n_-)\}$ and for a finite dimensional vector space $W_n$, $\Gamma'_{\Xi_n}(M^n,W_n)$ is the subspace of $\Gamma'(M^n,W_n)$ consisting of distributions with wave front set contained in the open cone $\Xi_n$. The first three factors in the space $W_n$ correspond accordingly to the Lagrange multipliers, antighosts and ghosts. The subsequent factors are the antifields.

The fact that gauge fixing is possible implies that if we keep all the unphysical fields with fixed values, then the initial value problem is well posed for the physical fields. This is however not enough. Since the differential $\delta^g$ is the Koszul map for the gauge fixed action, we should understand the equations of motion as equations for the full field multiplet with the auxiliary fields included. The functional derivative of a function $F\in\Ci_\mc(\E(M),\A)$ at the point $A_0\in\E(M)$ is an $\A$-valued distribution: $F^{(1)}_A(A_0):=\frac{\delta F}{\delta A}(A_0)\in \E'(M)\widehat{\otimes}\A$. The details on the involved topologies are given in Appendix \ref{topo}. The derivatives with respect to the odd variables are defined pointwise.  For example for the functions of ghost fields we have: $F^{(1)}_C(A_0):= \big(F(A_0)\big)^{(1)}_C$, $F(A_0)\in\A$, where the derivative  on the graded algebra $\A$ is defined as in \cite{Rej}:
\be\label{d1}
F^{(1)}(a)[h]:=F(h\wedge a)\quad F\in \La^p{\frakgo}'(M),\ a\in\La^{p-1}\frakgo(M),\ h\in\frakgo(M)\ p>0
\ee
Note that $F^{(1)}_C(A_0)\in{\frakgo}'(M)\widehat{\otimes}\A(M)$. Now to implement the equations of motion we take the quotient of $\BV(M)$ by the ideal generated by graded functions of the form:
\be\label{eomf}
A_0\mapsto\langle {S^g}_\alpha^{(1)}(A_0), \beta(A_0)\rangle\,,
\ee
where $A_0\in\E(M)$, $\alpha$ is  $A,B,C$ or $\bar{C}$ and $\beta(A_0)$ is the appropriate test section. Note that in the graded case, when $S^g$ is of degree higher than $1$ in anticommuting variables, we don't have an interpretation of the equations of motion as equations on the configuration space. Instead the algebraic definition on the level of functionals can still be applied \cite{Rej}. We can compare this situation to the purely bosonic case, when we had to show that the ideal $\F_0(M)$ is generated by the equations of motion for a given action functional. In the fermionic case we reason in the opposite direction and \textit{define} this ideal as generated by the elements of the form (\ref{eomf}). Equivalently we can say that it is the image of the map $\delta^g$ acting on $\#\ta=1$-grade derivations, so the on-shell functionals for the action $S^g$ are indeed characterized by $H_0(\A(M),\delta^g)$.

We conclude that after the gauge fixing the full dynamics is described by the action $S^g$ and therefore this generalized Lagrangian is the starting point for the construction of the Peierls bracket. The off-shell formalism \cite{DF04,DF02} is to be understood with respect to $S^g$ and going on-shell means taking the quotient by the ideal generated by the equations of motion. The construction of the Peierls bracket is a straightforward generalization of the construction done in the scalar case \cite{BDF,BFR}. The only subtle point is the grading. We discussed the general case of Peierls bracket for anticommuting fields in \cite{Rej}. All the distributional operations have to be generalized to distributions with values in a graded algebra. This can be easily done, because we equipped $\A(M)$ with a nuclear, sequentially complete topology. The distributional operations like convolution, contraction and pointwise product generalized to the $\A(M)$-valued distributions are now graded commutative (see the Appendix \ref{vvalued} and \cite{Rej} for details). As in the case of $\RR$-valued distributions, the pointwise product is well defined only when the sum of the wave front sets of the arguments does not intersect the zero section of the cotangent bundle. We point out that the use of $\A(M)$-valued distributions already accounts for the grading, so there is no need to introduce additional Grassman algebras by hand.

Since $S^g$ has at most quadratic terms with respect to the anticommuting variables, its second derivative can be again treated as an operator on the extended configuration space \cite{Rej}. To construct the Peierls bracket we need this operator to be strictly hyperbolic. Therefore we need to find a gauge fixing Fermion which makes the linearized equations of motion of $S^{g}$ into a strictly hyperbolic system in variables $A,B,C$ and $\bar{C}$. The existence of such a Fermion in a general case is an open question. In case of Yang-Mills theory it suffices to take the Lorenz gauge with $\Psi$ given by (\ref{Lorenz}). Taking the first functional derivative of $S^g$ results in a following system of equations\footnote{These equations should be understood as relations in the algebra $\A(M)$, that we have to quotient out. For example (\ref{1YM}) means that we quotient out the ideal generated (in the algebraic and topological sense) by evaluation functionals $(\hinv D\!*\!F-dB-i[d\bar{C}, C])^I_\mu(x)$.\label{eqs}}:
\begin{eqnarray}
\hinv D\!*\!F=\hinv D\!*\!DA&=&-dB-i[d\bar{C}, C]\label{1YM}\,,\\
\hinv d\!*\!A+\alpha B&=&0\,,\label{gaugecond}\\
\hinv d\!*\!DC&=&0\,,\nonumber\\
\hinv D\!*\!d\bar{C}&=&0\,,\nonumber
\end{eqnarray}
where $D\omega=D+[A,\omega]$ denotes the covariant derivative. Acting with $\hinv D*$ on equation (\ref{1YM}) we obtain an evolution equation for $B$:
\begin{equation*}
\hinv D*dB=-i*[d\bar{C}, *DC]\,.
\end{equation*}
For every field from configuration space the second variational derivative of (\ref{fixed}) is an  integral kernel of a normal hyperbolic differential operator. Indeed, in the linearized system of equations the only terms containing second derivatives in (\ref{1YM}) are of the form $\hinv d*dA=\Box A-d\hinv d*A$. From the gauge fixing condition (\ref{gaugecond}) it follows that $\hinv d*A=-\alpha B$ and therefore the only contributions containing the second derivatives are of the form $\Box \phi^\alpha$, where $\phi^\alpha=A,C,\bar{C}$ or $B$. This means that ${S^g}''_M$ provides a hyperbolic system of equations and one can construct the advanced and retarded Green's functions $\Delta^R_{S^g}$, $\Delta^A_{S^g}$. 
We define the Peierls bracket by:
\begin{eqnarray}
\{F,G\}_{S^g}&\doteq& R_{S^g} (F,G)-A_{S^g}(F,G)\,,\label{peierls}\\
R_{S^g} (F,G)&\doteq&\sum\limits_{\alpha,\beta}(-1)^{(|F|+1)|\phi_\alpha|}\left<F^{(1)}_\alpha,(\Delta^R_{S^g})_{\alpha\beta}*G_\beta^{(1)}\right>\,,\label{ret}\\
A_{S^g} (F,G)&\doteq&\sum\limits_{\alpha,\beta}(-1)^{(|F|+1)|\phi_\alpha|}\left<F^{(1)}_\alpha,(\Delta^A_{S^g})_{\alpha\beta}*G_\beta^{(1)}\right>\,,\label{adv}
\end{eqnarray}
where $\Delta^{R/A}_{S^g}$ has to be understood as a matrix, $\phi^\alpha=A,C,\bar{C}$ or $B$ and $|.|$ denotes the ghost number. The sign convention chosen here comes from the fact that we use only left derivatives. One can show that $\{.,.\}_{S^g}$ is a well defined graded Poisson bracket on $\A(M)$. Moreover the algebra  $\A(M)$ is closed under this bracket. The next proposition shows that there is a relation between this dynamical structure and the BRST symmetry.
\begin{prop}
The BRST operator $\gamma^g$ satisfies the graded Leibniz rule with respect to the Peierls bracket:
\begin{equation}
\gamma^g\{F,G\}_{S^g}=(-1)^{|G|}\{\gamma^gF,G\}_{S^g}+\{F,\gamma^gG\}_{S^g}\,.\label{BRSideal}
\end{equation}
\end{prop}
\begin{proof}
From the definition of the BRST operator we know that it is a graded derivation on the algebra $\BV(M)$, acting from the right. Therefore it holds:
\begin{multline*}
\gamma^g\left<F^{(1)}_\alpha,({\Delta^R_{S^g}}_{\alpha\beta})*G_\beta^{(1)}\right>=(-1)^{|\phi_\alpha|+|G|}\left<\gamma^g\left(F^{(1)}_\alpha\right),({\Delta^R_{S^g}}_{\alpha\beta})*G_\beta^{(1)}\right>+\\
+(-1)^{|\phi_\beta|+|G|}\left<F^{(1)}_\alpha,\left(\gamma^g({\Delta^R_{S^g}})_{\alpha\beta}\right)*G_\beta^{(1)}\right>+\left<F^{(1)}_\alpha,({\Delta^R_{S^g}})_{\alpha\beta}*\gamma^g\left(G_\beta^{(1)}\right)\right>\,.
\end{multline*}
Now, using the fact that $S^g$ is $\gamma^g$-invariant (follows from the nilpotency of $\tilde{s}$), we obtain:
\[
\gamma^gR_{S^g}(F,G)=(-1)^{|G|}R_{S^g}(\gamma^gF,G)+R_{S^g}(F,\gamma^gG)\,.
\]
The same holds for $A_{S^g}(F,G)$, so the result follows from the definition (\ref{peierls}).
\end{proof}
Now we want to show that the Peierls bracket $\{.,.\}_{S^g}$ is well defined on the algebra of gauge invariant observables.
We recall that $\F^{\inv}_S(M)\cong H^0(H_0(\BV(M),\delta^g),\gamma^g)$ and $H_0(\delta^g,\BV(M))$ is the on-shell algebra of the gauge-fixed action $S^g$. It was proved in \cite{BFR} that for the scalar field the subalgebra of 
functionals that vanish on-shell is a Poisson ideal with respect to $\{.,.\}_S$. A similar reasoning can be applied also to the graded case and one shows that the image of $\delta^g$ in degree $\#\ta=0$ is a Poisson ideal with respect to
$\{.,.\}_{S^g}$. This means that the Peierls bracket is well defined on-shell, i.e. on the homology classes  $H_0(\delta^g,\BV(M))$. To see that it is also compatible with the differential $\gamma^g$ we consider $F,G\in \ke(\gamma^g)$ and from (\ref{BRSideal}) we conclude that
\begin{multline}\label{Poii}
\{F,G+\gamma^g H\}_{S^g}=\{F,G\}_{S^g}+\{F,\gamma^gH\}_{S^g}=\\
=\{F,G\}_{S^g}+\gamma^g\{F,H\}_{S^g}-(-1)^{|H|}\{\gamma^gF,H\}_{S^g}\\
=\{F,G\}_{S^g}+\gamma^g\{F,H\}_{S^g}\,.
\end{multline}
This shows that $\{.,.\}_{S^g}$ is compatible with the cohomology classes of $\gamma^g$. Using this result and the previous one, concerning the $0$-th homology of $\delta^g$, we conclude that the Peierls bracket is well defined on $\F^{\inv}_S(M)$. As a final remark, we note that for Yang-Mills theories the Poisson structure on $\F^{\inv}_S(M)$ defined by the gauge-fixed action doesn't depend on the gauge-fixing Fermion $\Psi$. Indeed, let $S^{g}_1=S+\gamma^g\Psi_1$, whereas  $S^{g}_2=S+\gamma^g\Psi_2$. Therefore $S^{g}_2=S^{g}_1+\gamma^g(\Psi_1-\Psi_2)$. It follows now that for $F,G\in \ke(\gamma^g)$ we have:
\begin{equation*}
\{F,G\}_{S^{g}_2}=\{F,G\}_{S^{g}_1}+\gamma^g(\ldots)\,.
\end{equation*}
It means that $\{F,G\}_{S^g_1}$ and $\{F,G\}_{S^g_2}$ are in the same cohomology class.

To end this section we discuss the functoriality of the construction presented above. Let $\PgAlg$ denote the category of graded topological Poisson algebras with continuous faithful graded Poisson algebra morphisms as morphisms. It can be shown that the assignment of $(\BV(M),\{.,.\}_{S^{g}})$ to $M$ is a covariant functor from $\Loc$ to $\PgAlg$. It is interesting to note, that on $\BV(M)$ we have two Poisson structures, one is the antibracket $\{.,.\}$ and the other one the Peierls bracket $\{.,.\}_{S^g}$. Both structures are natural and depend on the generalized Lagrangian $S$, defining the concrete classical theory. Moreover the algebra $\BV(M)$ is equipped with the graded differential $s$ and various gradings (antifield, ghost, \ldots). It would be interesting to investigate the relations between all these structures. 
\section{General relativity}
 %
 The treatment of quantum gravity in the framework of locally covariant quantum field theory was proposed in \cite{F,BF1}. The first step towards this program is the proper understanding of the structures that appear already on the classical level. This can be done using the BV-formalism described in the previous sections. It turns out that the direct application of the results known from the Yang-Mills theories is impossible, since there are no local diffeomorphism (i.e. gauge) invariant observables and the set $\F^\inv_S(M)$ would be trivial. We will show that the BV construction done for a fixed manifold $M$ leads to a trivial $0$-cohomology class. 
It turns out, however, that the framework of locally covariant field theory provides a solution. It was proposed in \cite{F,BF1} to consider gauge invariant \textit{fields} instead of observables. Here the fields are understood as natural transformations \cite{BFV,Few} between the functors $\E_c$ and $\F$. Let $\Nat( \E_c,\F )$ denote the set of natural transformations.
In analogy to Section \ref{CME} we define the set of fields as $\bigoplus\limits_{k=0}^\infty \Nat(\E_c^k,\F)$.
We will show that this is the right structure to consider as a starting point for the BV construction. Indeed in general relativity one always uses objects that are natural, for example the scalar curvature. Although it doesn't make sense to consider it at a fixed spacetime point, it is still meaningful to treat it as an object defined in all spacetimes in a coherent way. This is the underlying idea of identifying the physical quantities with natural transformations.
\subsection{BV construction on a fixed background}
We start with recapitulating the standard approach to the BRST construction made for general relativity. We shall perform it on a fixed background manifold $(M,g)$ and show that this leads to a trivial set of gauge invariant quantities.

For the classical gravity the configuration space is $\E(M)=(T^*M)^{2\otimes}\doteq T^0_2M$, the space of rank $(0,2)$ tensors. The Einstein-Hilbert action reads\footnote{In this chapter we use the metric signature $(-+++)$ and the conventions for the Riemann tensor agreeing with \cite{Wald}.}: 
\be\label{EH}
S_{(M,g)}(f)(h)\doteq \int R[g+h]f\,\textrm{d vol}_{(M,g+h)}\,,
\ee
where $g$ is the background metric, $h$ the perturbation and $\tilde{g}=g+h$. For every $g$ the local functional $S_{(M,g)}(f)(h)$ is defined in some open neighborhood $U_g\subset\E(M)$. We can make this neighborhood small enough to guarantee that $\tilde{g}$ is a Lorentz metric with the signature $(-+++)$. Since we are interested only in the perturbation theory, we don't need $S_{(M,g)}(f)(h)$ to be defined on the full configuration space. The diffeomorphism invariance of (\ref{EH}) means that the symmetry group of the theory is the diffeomorphism group $\Diff(M)$. Since we are interested only in local symmetries, we can restrict our attention to $\Diff_c(M)$. It is an infinite dimensional Lie group modeled on $\X_c(M)$, the space of compactly supported vector fields on $M$ \cite{Michor80,Michor,Gloe06,Gloe07}. We can now define the action of $\Diff_c(M)$ on $\E(M)$ or more generally on arbitrary tensor fields. Let $\Tens(M)$ denote the space of smooth sections of the vector bundle $\bigoplus\limits_{k,l}T^k_lM$, where $T^k_lM\doteq \underbrace{TM\otimes\ldots\otimes TM}_{k}\otimes \underbrace{T^*M\otimes\ldots\otimes T^*M}_{l}$. We define a map $\rho_M:\Diff_c(M)\rightarrow L(\Tens(M),\Tens(M))$ as a pullback, namely:
\begin{equation}
\rho_M(\phi)=(\phi^{-1})^*t,\quad \phi\in \Diff_c(M),\ t\in\Tens(M)
\end{equation}
The algebra of compactly supported vector fields $\X_c(M)$ on $M$ is the Lie algebra of $\Diff_c(M)$, and there exists an exponential mapping. The corresponding derived representation of $\X_c(M)$ on $\Tens(M)$  is just the Lie derivative:
\begin{equation}
\rho_M(X)t\doteq\frac{d}{dt}\Big|_{t=0}(\exp(-tX))^*t=\pounds_{X}t,\label{repre}
\end{equation}
where $X\in\X_c(M)$ and the last equality follows from the fact that the exponential mapping of the diffeomorphism group is given by the local flow. The most general nontrivial symmetries of the action (\ref{EH}) can be written as elements of $\Ci_\ml(\E(M),\X_c(M))$. Like in gauge theories one can define the action of $\X_c(M)$ on $\F(M)$, the space of functionals on the configuration space. It is given by: 
\be\label{rhoM3}
\partial_{\rho_M(X)}F(h)=\left<F^{(1)}(h),\pounds_{X}\tilde{g}\right>,\qquad F\in\F(M), X\in\X_c(M)\,.
\ee
The full BV-complex is defined by:
\be
\BV(M)=\Ci_\ml\big(\E(M),\La\E_c(M)\widehat{\otimes}\La{\Xo}'(M)\widehat{\otimes}S^\bullet \X_c(M)\big)\label{BVfix}
\ee
The BV differential can be now defined analogously to the gauge theories case as $s_0F=\{F,S_M(f)+\theta_M(f)\}$
where $f\equiv1$ on $\supp\, F$, $F\in\BV(M)$ and $\theta_M$ is constructed form the representation $\rho_M$. It can be seen already at this point that $H^0(\BV(M),s_0)$ is trivial, because there are no compactly supported diffeomorphism invariant functionals (see for example \cite{Weise} for a detailed discussion).
\subsection{BV construction for natural transformations}\label{ntBV}
Now we can come back to the question, how to define the BV differential on the level of fields, i.e. on $\bigoplus\limits_{k=0}^\infty \Nat(\E_c^k,\F)$. We have to find out how the natural transformations are transforming under spacetime diffeomorphisms. First we note that if $\alpha:M\rightarrow N$ is an isometric diffeomorphism, then the naturality condition implies that:
\[
\Phi_M(f)(\alpha^*h)=\Phi_{\alpha(M)}(\alpha_*f)(h)\,,
\]
where $f\in\E_c(M)$ and $h\in\E(\alpha(M))$. Therefore the basic consistency condition to impose on the action of symmetries on the natural transformations is:
\be
\rho_M(\alpha)\Phi_{(M,g)}=\Phi_{(\alpha(M),\alpha_*g)}\,,\label{consis}
\ee
where $\alpha\in\Diff_c(M)$. The infinitesimal version of this condition leads to the following action of symmetries on natural transformations:\footnote{Here we defined $\Xo$ to be a contravariant functor, using the fact that vector fields can be mapped to 1-forms using the metric and the forms can be subsequently pulled back by isometric embeddings.}
\begin{eqnarray}
(\rho_M(X)\Phi_M)(f)&:=&\partial_{\rho_M(X)}(\Phi_M(f))+\Phi_M(\rho_M(X)f)\label{repgrav}\\
&=&\partial_{\rho_M(X)}(\Phi_M(f))+\Phi_M(\pounds_Xf),\qquad X\in\Xo(M)\,.\nonumber
\end{eqnarray}
In other words, we first act with the representation of $X$ on the functional in $\F(M)$ and then on the test field $f\in\E_c(M)$. In the above formula $\rho(.)\Phi$ is a natural transformation between the functors $\D$ and $\Ci_\ml(\E,\Xo)$. 
\footnote{A related discussion of the proper choice of the BRST transformation for general relativity can already be found in a paper of Nakanishi \cite{Naka} (cf. also \cite{NaOji}).}

Once we identified the physical quantities with fields, the condition (\ref{consis}) already distinguishes the action (\ref{repgrav}) as the right starting point for the BV construction. The Chevalley-Eilenberg complex on the level of natural transformations is defined as $\bigoplus\limits_{k=0}^\infty \Nat(\E_c^k,\CE)$, where $\CE(M)$ is constructed analogously as in Yang-Mills theories. We can define the Chevalley-Eilenberg differential $\gamma=\gamma^{(0)}+\gamma^{(1)}$ by:
\begin{eqnarray*}
\gamma^{(0)}\Phi&=&\rho(.)\Phi,\qquad \Phi\in\Nat(\E_c,\F)\,,\\
\gamma^{(1)}\Phi&=&-\Phi\circ[.,.],\qquad\Phi\in\Nat(\E_c,\Xo')\,.
\end{eqnarray*}
Requirement of the graded Leibniz rule allows us to extend $\gamma$ to the whole $\bigoplus\limits_{k=0}^\infty \Nat(\E_c^k,\CE)$.
In case of general relativity it is convenient to choose as a space of test fields the space of compactly supported tensor fields $\Tens_c(M)$. Now we can repeat the construction of the BV complex on the level of natural transformations. We define the extended algebra of fields as:
\[
Fld=\bigoplus\limits_{k=0}^\infty \Nat(\E_c^k,\BV)\,,
\]
with $\BV(M)$ given by (\ref{BVfix}).
The set $Fld$ becomes a graded algebra if we equip it with a graded product defined as:
\be\label{ntprod}
(\Phi\Psi)_M(f_1,...,f_{p+q})=\frac{1}{p!q!}\sum\limits_{\pi\in P_{p+q}}\mathrm\Phi_M(f_{\pi(1)},...,f_{\pi(p)})\Psi_M(f_{\pi(p+1)},...,f_{\pi(p+q)})\,,
\ee
where the product on the right hand side is the product of the algebra $\BV(M)$.
We can also introduce on $Fld$ a graded bracket using definition (\ref{ntbracket}). This bracket is graded antisymmetric, satisfies the graded Jacobi identity (\ref{Jacid}) and the Leibniz rule (\ref{leibniz}) with respect to the product (\ref{ntprod}), so $(Fld,\{.,.\})$ is a graded Poisson algebra.
The $BV$-differential on $Fld$ is now given by:
\[
(s\Phi)_M(f):=s_0(\Phi_M(f))+(-1)^{|\Phi|}\Phi_M(\rho_M(.)f)\,,
\]
where $s_0$ is the differential defined in the previous section.
The $0$-cohomology of $s$ is nontrivial, since it contains for example the Riemann tensor contracted with itself, smeared with a test function:
\[
\Phi_{(M,g)}(f)(h)=\int\limits_M R_{\mu\nu\alpha\beta}[\tilde{g}]R^{\mu\nu\alpha\beta}[\tilde{g}] fd\textrm{vol}_{(M,\tilde{g})}\qquad \tilde{g}=g+h\,.
\]
With the general framework proposed above we can now treat more specific problems in general relativity. We claim that the physical quantities should be identified with the elements of $H^0(Fld,s)$. As natural transformations they define what it means to have the same physical objects in all spacetimes. In this sense we get a structure that is completely covariant. We can now introduce dynamics on $H^0(Fld,s)$ by defining the Poisson bracket. To have a better control on the wavefront set we shall use again the topology $\tau_\Xi$ and replace the multilocal with the microcausal functionals. The corresponding space of natural transformation is denoted by $Fld_{mc}$. We can construct the Peierls bracket analogously as in Section \ref{Peierls}. Our starting point is a Lagrangian that implements the differential $s_0$:
\be
L_M(f):=\int fR[\tilde{g}]\,d\textrm{vol}_{(M,\tilde{g})}+\int\!\!d\textrm{vol}_{(M,g)}( f\pounds_C \tilde{g}_{\mu\nu})\frac{\delta}{\delta h_{\mu\nu}}+\frac{1}{2}\int\!\!d\textrm{vol}_{(M,g)}(f[C,C]^\mu)\frac{\delta}{\delta C^{\mu}}\,.
\ee
To impose the gauge fixing we introduce the nonminimal sector. We shall do it already on the level of natural transformations. The functions of Nakanishi-Lautrup fields will be the elements of $\Nat(\E_c,S^\bullet\X')$ and functions of antighosts will belong to $\Nat(\E_c,\La\X')$. We can define the BV operator on the nonminimal sector simply as: $s\Phi_1:=\Pi\Phi_1\circ m_i$, $s\Phi_2=0$ for $\Phi_1\in \Nat(\E_c,\La^1\X')$, $\Phi_2\in \Nat(\E_c,S^1\X')$. To impose the gauge fixing we shall use a gauge fixing Fermion $\Psi\in Fld_{mc}$. It induces a transformation of $Fld_{mc}$ given by (\ref{gfermion}). This transformation is an isomorphism on the cohomology groups, since:
 \[
 (\tilde{s}\tilde{X})_M(f_1)=\widetilde{\{X_M(f_1),L_M(f_2)\}}+(-1)^{|X|}\tilde{X}_M(\rho_M(.)f_1)= \widetilde{(sX)}_M(f_1)\,,
 \]
where $X$ is a natural transformation with values in derivations, $\tilde{X}:=\alpha_\Psi(X)$ and $f_2\equiv 1$ on the support of $f_1$. The above result can be written more compactly as:
 \[
 \tilde{s}\tilde{X}=\widetilde{sX}
 \]
 To fix the gauge we have to choose $\Psi$. Since it has to be covariant, the most natural choice is the background gauge (see \cite{Ichi,NiOk}), i.e.:
 \[
 \Psi_{(M,g)}(f)=i\int\!\!\textrm{dvol}_{(M,g)}\left(\frac{\alpha}{2}\bar{C}_\mu B^\mu+\frac{1}{\sqrt{-g}}\bar{C}_\mu\nabla_\nu (\gt{\nu}{\mu})\right)=i\int\!\!g\big(\bar{C}, \frac{\alpha}{2}B+K(h)\big)\textrm{dvol}_{(M,g)}\,,
 \]
 where the indices are lowered in the background metric $g$, $\nabla$ is the covariant derivative on $(M,g)$ and we denoted $\tilde{\textfrak{g}}^{\nu\lambda}:=\sqrt{-\tilde{g}}\tilde{g}^{\nu\lambda}$, $K^\mu(h)=\frac{1}{\sqrt{-g}}\nabla_\nu \gt{\nu}{\mu}$. For $\alpha=0$ this is just the harmonic gauge. After putting antifields to $0$ we obtain a following form of the gauge-fixed Lagrangian:
 \[
 L^g_M(f):=\int fR[g+h]\,\textrm{d vol}_{(M,g+h)}+\int\!\!\textrm{dvol}_{(M,g)}f\left(ig\big(\bar{C},\frac{\delta K}{\delta h}[\pounds_C \tilde{g}]\big)-g\big(B, \frac{\alpha}{2}B+K(h)\big)\right)\,.
 \]
The differential $\tilde{s}$  can be expanded with respect to the total antifield number as $\tilde{s}=\delta^g+\gamma^g$, where $\delta^g$ is the Koszul differential of the gauge fixed action and $\gamma^g$ is the gauge-fixed BRST differential given by:
 \[
(\gamma^g\Phi)_M(f)=\gamma_0^g(\Phi_M(f))+(-1)^{|\Phi|}\Phi_M(\rho(.)f)\,,
 \]
 where:
 \begin{center}
{\setlength{\extrarowheight}{2.5pt}
\begin{tabular}{ll}
\toprule%
& $\gamma_0^g$\\\otoprule%
$F\in\F(M)$&$\left<F^{(1)},\pounds_C(g+.)\right>$\\
 $C$&$-\frac{1}{2}[C,C]$\\
 $B$& $0$\\
  $\bar{C}$& $iB$\\\bottomrule
\end{tabular}}
\end{center}
The algebra of physical microcausal fields can be recovered as $Fld_{ph}:=H^0(H_0(Fld_{mc},\delta^g),\gamma^g)$. To introduce the Poisson structure on it, we shall first do it on $Fld_{mc}$. We start with finding the field equations for the action $S^g$. We use the fact that in local coordinates:
 \[
 \pounds_C\gt{\mu}{\nu}=-\gt{\alpha}{\nu}\nabla_\al C^\mu-\gt{\alpha}{\mu}\nabla_\al C^\nu+\nabla_\alpha (C^\alpha\gt{\mu}{\nu})\,.
\]
The field equations\footnote{The field equations have to understood in the algebraic sense, see footnote \ref{eqs}.} take the form (compare with the equations in \cite{NiOk} obtained for the Minkowski background):
\begin{eqnarray*}
R_{\nu\lambda}[\tilde{g}]&=&-i\Big(\nabla_{(\nu}\bar{C}_{|\mu|}\nabla_{\la)}C^\mu+\nabla_\mu\bar{C}_{(\nu}\nabla_{\la)}C^\mu+(\nabla_\al\nabla_{(\nu}\bar{C}_{\la)})C^\al\Big)-\nabla_{(\la}B_{\nu)}\\%
\tilde{\textfrak{g}}^{\nu\lambda}\nabla_{\nu}\nabla_{\lambda} C^\mu&=&\tilde{\textfrak{g}}^{\alpha\nu}R_{\lambda\nu\alpha}^{\quad\ \mu}[g]C^\lambda+\alpha\sqrt{-g}\Big(B^\lambda\nabla_\lambda C^\mu-\nabla_\lambda(B^\mu C^\lambda)\Big)\\%
\tilde{\textfrak{g}}^{\nu\lambda}\nabla_{\nu}\nabla_{\lambda} \bar{C}_\mu&=&-\tilde{\textfrak{g}}^{\nu\lambda}R_{\nu\mu\lambda}^{\quad\ \alpha}[g]\bar{C}_\alpha+\alpha\sqrt{-g}\Big(B^\lambda\nabla_\lambda \bar{C}_\mu+B^\lambda \nabla_\mu \bar{C}_\lambda\Big)
\\%
\nabla_\nu \tilde{\textfrak{g}}^{\nu\mu}&=&-\alpha\sqrt{-g} B^\mu
\end{eqnarray*}
This system is gauge-fixed but not strictly hyperbolic in all the variables, since we have second derivatives of the ghosts in the first equation. This can be fixed by a suitable variable change. Before setting antifields to $0$ we perform a canonical transformation of the algebra $\BVn(M)$ by setting $b_\la=B_{\la}-iC^\al\nabla_{\al}\bar{C}_{\la}$. The antifields have to transform in such a way that the antibracket remains conserved, i.e.: $\frac{\delta}{\delta c_\la}=\frac{\delta}{\delta C_\la}+i\nabla^{\la}\bar{C}_{\bet}\frac{\delta}{\delta B_\bet}$ and $\frac{\delta}{\delta \bar{c}_\la}=\frac{\delta}{\delta \bar{C}_\la}-iC^\al\nabla_{\al}(.)\circ\frac{\delta}{\delta B_\la}$. All other variables remain unchanged. The new gauge-fixed Lagrangian takes the form:
 \begin{multline}
 L^g_M(f):=\int fR[\tilde{g}]\,d\textrm{vol}_{(M,\tilde{g})}+\\
 +\int\!\!d\textrm{vol}_{(M,g)}f\left(ig\big(\bar{c},\frac{\delta K}{\delta h}[\pounds_c \tilde{g}]\big)-g\Big(b+(ic^\al\nabla_{\al})\bar{c}, \frac{\alpha}{2}(b+(ic^\al\nabla_{\al})\bar{c})+K(h)\Big)\right)
 \end{multline}
The gauge-fixed BRST differential $\gamma^g_0$ is now defined as:
 \begin{center}
{\setlength{\extrarowheight}{2.5pt}
\begin{tabular}{ll}
\toprule%
& $\gamma_0^g$\\\otoprule%
$F\in\F(M)$&$\left<F^{(1)},\pounds_c(g+.)\right>$\\
 $c$&$-\frac{1}{2}[c,c]$\\
 $b$& $i(c^\bet\!\wedge c^\al\nabla_\bet\nabla_\al)\bar{c}+c^\al\nabla_\al b$\\
  $\bar{c}$& $ib-c^\la\nabla_\la \bar{c}$\\\bottomrule
\end{tabular}}
\end{center}
The equations of motion in the new variables can be written as:
\begin{eqnarray}
\tilde{R}_{\nu\lambda}&=&-i\nabla_{(\nu}\bar{c}_{|\mu|}\nabla_{\la)}c^\mu-\nabla_{(\nu}b_{\la)}+iR_{\al\nu\bet\la}\bar{c}^{(\bet} c^{\al)}\label{sys2}\,,\\%
\tilde{\textfrak{g}}^{\nu\lambda}\nabla_{\nu}\nabla_{\lambda} c^\mu&=&\tilde{\textfrak{g}}^{\alpha\nu}R_{\lambda\nu\alpha}^{\quad\ \mu}c^\lambda+\alpha\sqrt{-g}\Big((b^\la+ic^\al\nabla_\al \bar{c}^\la)\nabla_\lambda c^\mu-\nabla_\lambda((b^\mu+ic^\al\nabla_\al \bar{c}^\mu) c^\lambda)\Big)\,,\nonumber\\%
\tilde{\textfrak{g}}^{\nu\lambda}\nabla_{\nu}\nabla_{\lambda} \bar{c}_\mu&=&-\tilde{\textfrak{g}}^{\nu\lambda}R_{\nu\mu\lambda}^{\quad\ \alpha}\bar{c}_\alpha+\alpha\sqrt{-g}\Big((b^\la+ic^\al\nabla_\al \bar{c}^\la)\nabla_\lambda \bar{c}_\mu+(b^\la+ic^\al\nabla_\al \bar{c}^\la) \nabla_\mu \bar{c}_\lambda\Big)\,,\nonumber
\\%
\nabla_\nu \tilde{\textfrak{g}}^{\nu\mu}&=&-\alpha\sqrt{-g} (b^\mu+ic^\al\nabla_\al \bar{c}^\mu)\,,\nonumber
\end{eqnarray}
where we denoted $\tilde{R}_{\nu\la}:=R[\tilde{g}]_{\nu\la}$ and $R_{\al\bet\gamma\la}:=R[g]_{\al\bet\gamma\la}$. The equation for $b$ can be obtained from the first equation by means of the Bianchi identity. One can already see that after linearization we obtain a strictly hyperbolic system of equations since: $c^\la c^\al\partial_\lambda\partial_\al \bar{c}^\mu=0$ and all the other second order terms are of metric type.
For such a system retarded and advanced solutions of the linearized equations exist and one can define the Peierls bracket on $Fld_{mc}$. Like in case of Yang-Mills theories it is well defined also on $Fld_{ph}$ and we obtain a Poisson algebra $(Fld_{ph},\{.,.\}_{S^g})$. Although the Poissoin structure on $Fld_{mc}$ can depend on the choice of variables in the extended algebra, this doesn't affect the structure induced on $Fld_{ph}$. Note that  for the harmonic gauge ($\al=0$) and the Minkowski background the system (\ref{sys2}) simplifies to (compare with \cite{NaOji}):
\begin{eqnarray*}
\tilde{R}_{\nu\lambda}&=&-i\partial_{(\nu}\bar{c}_{|\mu|}\partial_{\la)}c^\mu-\partial_{(\la}b_{\nu)}\\%
\Box_{\tilde{g}}c^\mu&=&0\\%
\Box_{\tilde{g}}\bar{c}_\mu&=&0\\%
\Box_{\tilde{g}}b_\mu&=&0\\%
\partial_\nu \tilde{\textfrak{g}}^{\nu\mu}&=&0
\end{eqnarray*}
%
\section{Conclusions}
%
In this paper we developed the BV formalism for  locally covariant classical
field theory. We showed that the structure can be understood in terms of the geometry of the configuration space which, in typical cases, is the space of sections of a bundle over a globally hyperbolic spacetime. We took seriously the facts that the configuration space is an infinite dimensional differential manifold, modeled over a suitable locally convex space, and that the underlying spacetimes are never compact. We also did not restrict ourselves to spacetimes with compact Cauchy surfaces.

We analyzed in details the topological and functional analytic aspects of the construction. To achieve this we used the methods of calculus on locally convex vector spaces. This research area in mathematics undergoes now a very dynamical development and there are many interesting results that can be applied in physics. In the future we wish to investigate these issues in more detail. Up to now the mathematically rigorous treatment of the BV formalism was done mainly in a purely algebraic framework, when the topological aspects are neglected. 

It turned out to be crucial that all constructions are functorial. In particular the actions are not elements of the algebra of observables over a given spacetime, instead they have to be considered as natural transformations between suitable functors. The BRST transformations are locally, but not globally inner derivations of the graded Poisson algebra  of the BV complex. We also showed how the classical master equation 
can be formulated on the level of natural transformations. We want to stress that the requirement of covariance with respect to the isometric spacetime embeddings automatically forces the extended action to be a natural transformation rather than an element of the algebra of functionals. This is again an indication that the structures we are using are natural and provide the right formulation of classical field theory.

The thorough distinction between local and global aspects and the emphasis on functoriality pays off in the case of general relativity. There the BV complex for a fixed spacetime turns out to have trivial cohomology, in agreement with the nonexistence of local observables in gravity. But on the level of locally covariant fields, considered as suitable natural transformations between functors on the underlying category of spacetimes, the BV complex has a nontrivial cohomology which contains the expected observables, i.e. curvature and the related quantities. The theory is no longer a theory on a fixed spacetime, but depends only on the chosen class of spacetimes.

We hope that the framework we developed will provide a basis for a conceptually consistent approach to quantum gravity.
\section*{Acknowledgements}
We would like to thank Ch. Wockel, R. Brunetti, P. Lauridsen Ribeiro and J.-Ch. Weise for enlightening discussions and comments. K. R. is also grateful to the Villigst Stiftung for the financial support of the Ph.D.    
\appendix
%
\section{Appendix}
\subsection{Calculus on locally convex vector spaces}\label{iddg}
Let $X$ and $Y$ be topological vector spaces, $U \subseteq X$ an open set and $f:U \rightarrow Y$ a map. The derivative of $f$ at $x$ in the direction of $h$ is defined as
\be\label{de}
df(x)(h) \doteq \lim_{t\rightarrow 0}\frac{1}{t}\left(f(x + th) - f(x)\right)
\ee
whenever the limit exists. The function f is called differentiable at $x$ if $df(x)(h)$ exists for all $h \in X$. It is called continuously differentiable if it is differentiable at all points of $U$ and
$df:U\times X\rightarrow Y, (x,h)\mapsto df(x)(h)$
is a continuous map. It is called a $C^1$-map if it is continuous and continuously differentiable. Higher derivatives are defined for $C^n$-maps by 
\be
d^n f (x)(h_1 , \ldots , h_n ) \doteq \lim_{t\rightarrow 0}\frac{1}{t}\big(d^{n-1} f (x + th_n )(h_1 , \ldots, h_{n-1} ) -
 d^{n-1}f (x)(h_1 , \ldots, h_{n-1}) \big)
 \ee
The derivative defined by (\ref{de}) has many nice properties. It is shown for example in \cite{Neeb,Ham}, that the following proposition is valid:
\begin{prop}
Let $X$ and $Y$ be locally convex spaces, $U \subseteq X$ an open subset, and $f:U \rightarrow Y$ a continuously differentiable function. Then:
\begin{enumerate}
\item	For any $x \in U$ , the map $df(x): X \rightarrow Y$ is real linear and continuous.
\item (Fundamental Theorem of Calculus). If $x + [0, 1]h \subseteq U$ , then 
\[
f (x + h) = f (x) +\int\limits_0^1 df (x + th)(h) dt\,.
\]
\item $f$ is continuous.
\item If $f$ is $C^n$, $n \geq 2$, then the functions $(h_1,...,h_n) \mapsto d^nf(x)(h_1,...,h_n)$, $x \in U$, are
symmetric $n$-linear maps.
\item If $x + [0, 1]h \subseteq U$,then we have the Taylor Formula:
\begin{multline*}
f (x + h) = f (x) + df (x)(h) + \ldots+\frac{1}{(n-1)!}d^{n-1} f (x)(h,\ldots, h)+\\
+\frac{1}{(n-1)!}\int\limits_0^1(1-t)^{n-1}d^nf(x+th)(h,...,h)dt\,.
\end{multline*}
\end{enumerate}
\end{prop}
Now, following \cite{Neeb} we shall introduce a notion of an infinite dimensional manifold.
Let $M$ be a Hausdorff topological space and $E$ a locally convex space. An $E$-chart of an open subset $U \subseteq M$ is a homeomorphism $\varphi:U \rightarrow \varphi(U) \subseteq E$ onto an open subset $\varphi(U)$ of $E$. We denote such a chart as a pair $(\varphi,U)$. Two charts $(\varphi,U)$ and $(\psi,V)$ are said to be smoothly compatible if the map
$\psi \circ \varphi^{-1}\Big|_{\varphi(U\cap V)}: \varphi(U\cap V ) \rightarrow \psi(U \cap V )$
is smooth.	\\
An $E$-atlas	of	$M$	is	a family	$(\varphi_i, U_i)_{i\in I}$	of	pairwise compatible $E$-charts of $M$ for which $\bigcup_i U_i=M$. Many of the objects used in differential geometry can be defined also in the infinite dimensional case. We start with the notion of a tangent space. 
\begin{df}
Let $a$ be an element of a locally convex vector space $X$. A kinematic tangent vector with foot point $a$ is a pair $(a, Q)$ with $Q\in X$. $T_a E \cong E$ is the space of all kinematic tangent vectors with foot point a. It consists of all derivatives $c'(0)$ at $0$ of smooth curves $c : \RR \rightarrow E$ with $c(0) = a$. The kinematic tangent space of a locally convex vector space $E$ will be denoted by $TE$ and the space of vector fields by $\Gamma(TE)$.
\end{df}
We use the term \textit{kinematic} since in the most general case this definition doesn't coincide with the definition of vector fields as derivations. Fortunately for the spaces considered in this paper this doesn't pose a problem.
\begin{df}
Let $M$ be a smooth manifold with the atlas $(\varphi_i, U_i)_{i\in I}$, where $\varphi_i:U_i\rightarrow E_i$. We consider the following equivalence relation on the disjoint union
\[
\bigcup\limits_{i\in I}U_i\times E_i \times \{i\}\,,
\]
\[
(x,v,i) \sim (y,w,j) \Leftrightarrow x = y\ \textrm{and}\ d(\ph_{ij})(\ph_j(x))w = v\,,
\]
where $\ph_{ij}$ are the transition functions. One denotes the quotient set by $TM$, the kinematic tangent bundle of $M$. 
\end{df}
Now we want to define differential forms on an infinite dimensional manifold. This turns out to be a problem, since there is no natural notion of a cotangent space (no natural topology on the dual of an infinite dimensional lcvs). There are alternative definitions of differential forms, but they are not equivalent. A detailed discussion of this problem is given in \cite{Michor} (VII.33). It turns out that only one of these notions is stable under Lie derivatives, exterior derivative, and the pullback. 
\begin{df}[\cite{Michor},VII.33.22]\label{33.22}
Let $M$ be a smooth infinite dimensional manifold. We will define
the space of differential forms on $M$ as:
\begin{equation*}
\Omega^k(M) \doteq \Ci(M \leftarrow L^k_\alte(TM,M \times \RR))\,.
\end{equation*}
Similarly, we denote by $\Omega^k(M;V) \doteq \Ci(M \leftarrow L^k_\alte(TM,M \times V))$
the space of differential forms with values in a locally convex vector space $V$.
\end{df}
\subsection{Topologies and completions}\label{topo}
In this section we shall give more details on the locally convex topologies appearing throughout the paper. Firstly we note that the configuration space $\E(M)$ can in our examples be equipped with a Fr\'echet topology, since it is just the space of smooth sections  $\E(M)=\Gamma(B)$ of some vector bundle $B\xrightarrow{\pi}M$ with a finite dimensional fiber $V$. The Fr\'echet topology is in this case generated by the family of seminorms:
\be
p_{K,m,a}(\ph)=\sup_{x\in K\atop |\alpha|\leq m}|\partial^\alpha\ph^a(x)|\,,
\ee
where $\alpha\in\NN^N$ is a multiindex and $K\subset M$ is a compact set. A set $B\subset\E(M)$ is bounded if $\sup_{\ph}\{p_{K,m,a}(\ph)\}<\infty$ for all seminorms $p_{K,m,a}$.
Let $\Bcal$ be the family of bounded sets in $\E(M)$. The \textit{strong topology} on the dual space $\E'(M)$ is defined by a family of seminorms: $p_B(T)\doteq\sup_{\ph\in B}\left<T,\ph\right>$, where $B\in \Bcal$, $T\in \E'(M)$, $\ph\in \E(M)$.

The space of compactly supported sections $\E_c(M)=\Gamma_c(E)$ can be equipped with a locally convex topology in a similar way. The fundamental system of seminorms is given by \cite{Sch0}:
\[
p_{\{m\},\{\epsilon\},a}(\ph)=\sup_\nu\big(\sup_{|x|\geq\nu,\atop |p|\leq m_\nu} \big|D^p\ph^a(x)\big|/\epsilon_\nu\big)\,,
\]
where $\{m\}$ is an increasing sequence of positive numbers going to $+\infty$ and $\{\epsilon\}$ is a decreasing one tending to $0$. This topology is no longer Fr\'echet. Nevertheless it possesses many nice properties, for example the spaces $\E(M)$, $\E_c(M)$, as well as their strong duals $\E'(M)$, $\E'_c(M)$, are reflexive nuclear spaces.  

Now we can finally address the question, what should be the natural topology on $\F(M)\subset\Ci(\E(M),\RR)$. There are many choices possible. For example one can take the topology of uniform convergance of each derivative on compact sets. This is discussed in details in \cite{Michor,Neeb,Ham}, but in our context the pointwise convergence of all the derivatives would suffice. 
This is the initial topology with respect to the mappings:
\be\label{tauF}
\F(M)\ni F\mapsto F^{(n)}(\ph)\in\Gamma'(M^{n},V^{\otimes n})\quad n\geq0\,,
\ee
where $\ph$ runs through all elements of $\E(M)$ and spaces $\Gamma'(M^{n},V^{\otimes n})$ are strong duals. We denote this initial topology on $\F(M)$ by $\tau$. Using the result of Pietsch (\cite{Pie}, 5.2.3) we conclude that this topology is nuclear.

Now we define the notion of locality for functions on $\E(M)$ with values in a locally convex vector space of distributional sections $\W(M):=\prod\limits_{n=0}^\infty\Gamma'(M^n,W_n)$ (equipped with the strong topology) for an arbitrary finite dimensional vector spaces $W_n$. In particular $\CE(M)$ and  $\fG(M)$ are embeded in  spaces of this form for $W_n=\Lambda^n g$ and $W_n=\bigoplus\limits_{k+l+m=n}\Lambda^k g\otimes \Lambda^l V\otimes S^mg$.
The $k$-th functional derivative of an element $F\in\Ci(\E(M),\W(M))$, at each point $\ph\in\E(M)$ is  vector valued compactly supported distribution (see \cite{Sch0,Sch1,Sch2,Kom} and Appendix \ref{vvalued}), i.e. an element of:
\[
F^{(k)}(\ph)\in\Gamma'(M^k,V^{\otimes k})\widehat{\otimes}\W(M)\cong \prod\limits_{n=0}^\infty\Gamma'(M^{k+n},V^{\otimes k}\otimes W_n)\,.
\]
On the right hand side of the above equation we have a direct product of distributions on $M$. We say that $F$ is local if all those distributions have their supports on the thin diagonal and their wavefront sets are orthogonal to the tangent bundles of the thin diagonals $\Delta^p(M)\doteq\left\{(x,\ldots,x)\in M^p:x\in M\right\}$, $p=n+k$, considered as subsets of the tangent bundles of $M^p$ (compare with (\ref{WFloc})). The subspace of $\Ci(\E(M),\W(M))$ consisting of all the local functions is denoted by $\Ci_\loc(\E(M),\W(M))$. The multilocal functions  $\Ci_\ml(\E(M),\W(M))$ are again sums of finite pointwise products of the local ones.
From the point of view of the Batalin-Vilkovisky formalism the topology $\tau$ is very well behaving and allows to perform the Koszul-Tate resolution for very general functional spaces. In principle one can replace the multilocal functions with functions that have first functional derivative as a smooth test section at each point of the configuration space. In this case the BV-construction can be performaed exactly as presented in sections \ref{scalar}-\ref{gaugefixing}. Unfortunately this class of functions is not stable
with respect to the Peierls bracket defined in Section \ref{Peierls}. To obtain a space that is closed under this bracket we need to extend $\F(M)$ to more singular objects. One also has to replace the topology $\tau$ with a different one. Since we need to have a control over the wave front sets, a natural choice is a topology introduced in \cite{BDF} and generalized for vector-valued functionals in \cite{Rej}.  
Let $\Xi_n$ denote the open cone $\{(x_1,...,x_n,k_1,...k_n)| (k_1,...k_n)\notin (\overline{V}_+^n \cup \overline{V}^n_-)\}$. Let $\Gamma'_{\Xi_n}(M^n,V^{\otimes n})$ be the subspace of $\Gamma'(M^n,V^{\otimes n})$ consisting of distributions with wave front set contained in the open cone $\Xi_n$.  Let now $C_n\subset \Xi_n$ be a closed cone contained in $\Xi_n$. We introduce (after \cite{Hoer,BaeF,BDF}) the following family of seminorms on $\Gamma'_{C_n}(M^n,V^{\otimes n})$: $p_{n,\phi,\tilde{C},k} (T) = \sup_{\xi\in V}\{(1 + |\xi|)^k |\widehat{\phi T}(\xi)|\}$,
where the index set consists of $(n,\phi,\tilde{C},k)$ such that $k\in \NN_0$, $\phi\in \D(U)$ and $\tilde{C}$ is a closed cone in $\RR^n$ with $(\supp ( \phi ) \times \tilde{C}) \cap C_n = \emptyset$. These seminorms, together with the seminorms of the weak topology provide a defining system for a locally convex topology deno ted by $\tau_{C_n}$. To control the wave front set properties inside open cones, we take an inductive limit. It can be shown that, to form this inductive limit one can choose the family of closed cones contained inside $\Xi_n$ to be countable. The resulting topology will be denoted by $\tau_{\Xi_n}$. The space of local vector-valued functionals
 can be now equipped with the initial topology analog to $\tau$, but with the topologies $\tau_{\Xi_n}$ on all the distribution and function spaces. We denote this topology by $\tau_{\Xi}$. It can be shown \cite{BDF} that each of the topologies $\tau_{C_n}$ is nuclear so $\tau_{\Xi}$ is nuclear as well. We define microcausal functionals $\F_\mc(M)$, as smooth compactly supported functions on $\E(M)$, for which the functional derivatives at each point are compactly supported distributions satisfying the microlocal spectrum condition:
 \be\label{mlsc}
 \WF(F^{(n)}(\ph))\subset \Xi_n, \qquad F\in\F_\mc(M)\,.
 \ee
Multilocal functionals are dense in $\F_\mc(M)$ with respect to the topology $\tau_{\Xi}$. In a similar way we define microcausal vector-valued functions. We say that $F\in\Ci(\E(M),\W(M))$ is microcausal if for all $\ph\in\E(M)$:
\[
F^{(k)}(\ph)\in\prod\limits_{n=0}^\infty\Gamma'_{\Xi_{n+k}}(M^{n+k},V^{\otimes k}\otimes W_n)\,.
\]
In particular the extended BV graded algebra $\BV_\mc(M)$ is defined to be a space of microcausal vector-valued functions. Now we want to extend the BV to those more singular objects.
Since the Lagrangian $L_M(f)$ is a local functional and its functional derivative is a smooth test section, the Koszul operator can be extended from multilocal vector fields to $\V_\mc(M):=\Ci_\mc(\E(M),\Gamma'_{\Xi_1}(M,V))$ and the resolution of  $\F_\mc(M)$ in the scalar case is provided by the differential graded algebra $\La\V_\mc(M)$. Similar reasoning applies also to the case when symmetries are present. The extended Chevalley-Eilenberg complex $\CE_\mc(M)$ consists of microcausal vector-valued functions with $W_n=\Lambda^n g$. Since the map $\partial_{\rho_M(.)}F$ is an element of 
 $\CE_\mc(M)$, for all $F\in\F_\mc(M)$, we can conclude that:
 \[
 H^0(\CE_\mc(M),\gamma)=\F^\inv_\mc(M)\,.
 \]
The full BV complex is equipped with the topology $\tau_{\Xi}$ and since multilocal functional lie dense in $\BV_\mc(M)$, the differential  $s$ can be extended to the full complex by continuity.
From the above discussion it follows that $H^0(\BV_\mc(M),s)$ is the space of microcausal gauge invariant functionals on-shell. We want to point out however, that the antibracket itself is not well defined on the whole $\BV_\mc(M)$. This is because the commutator of vector fields $\V(M)$ can be extended only to those elements of the space $\V_\mc(M)$, that have smooth first derivative. 
\subsection{Distributions with values in a locally convex vector space}\label{vvalued}
In this section we describe in details the theory of vector valued distributions in case when the vector space in question is a graded algebra $\A(M)$. For simplicity of notation we shall denote the graded product of $\A(M)$ by $\wedge$. We start with a definition \cite{Sch1}:
\begin{df}\label{vector}
Let $X$ be a locally convex topological vector space with a topology defined by a separable family of seminorms $\{p_\alpha\}_{\alpha\in I}$. We say that $T$ is a distribution on $\RR^n$ with values in $X$ if it is a continuous linear mapping from $\Dcal$ to $X$, where $\Dcal$ denotes the space of compactly supported functions on $\RR$.
\end{df}
In \cite{Sch1,Sch2} L. Schwartz requires the topological vector space $X$ to be quasicomplete.
\begin{df}[\cite{Jar}, 3.2]
A subset $U$ of a topological vector space $X$ is called \textit{complete} (\textit{sequentially complete}) if every Cauchy net (sequence) converges in $U$. We say that $X$ is quasi-complete if every closed bounded subset of $X$ is complete.
\end{df} 
The property of (sequential) completeness and quasi-completeness is inherited by the infinite direct products and infinite direct sums \cite[3.3.5]{Jar}.
\begin{df}
Let $E,F$ be Hausdorff lcvs and $\Bcal$ be the family of bounded sets of the completion of $E$ (bornology). Let $\tau_\Bcal$ be the topology of uniform convergence on bounded sets it induces on $L(E,F)$. We say that $E$ has the (sequential) approximation property if one of the following equivalent conditions holds:
\begin{enumerate}
\item $E'\otimes F$ is (sequentially) dense in $(L(E,F),\tau_\Bcal)$ for every $F$\,,
\item $E'\otimes E$ is (sequentially) dense in $(L(E,E),\tau_\Bcal)$\,,
\item $\1_E$ is the $\tau_\Bcal$-limit of some (sequence) net in $E'\otimes E$.
\end{enumerate}
\end{df}
The approximation property is also inherited by a direct product of a family of lcvs that are Hausdorff \cite[18.2.4]{Jar}. The spaces $\E(M)$, $\E_c(M)$, as well as their strong duals $\E'(M)$, $\E'_c(M)$, are complete and they have the approximation property (\cite{Jar,Sch1}) as well as the sequential approximation property\footnote{Unlike the approximation property used by L. Schwarz \cite{Sch1,Sch2}, the sequential approximation property doesn't follow from nuclearity. The first counterexample is due to \cite{Dub}, see also \cite{Vogt}.}. In \cite{Sch1,Sch2} it is shown that
for the quasi-complete $X$, the space of $X$-valued distributions with the topology of uniform convergance on compact sets is topologically isomorphic to the completed injective tensor product of $\Dcal'$ and $X$. It was discussed in \cite{Kom} that these results can be also applied to the situation when all the spaces are sequentially complete. In this paper we always deal with nuclear sequentially complete spaces $X$, so we can identify the space of distributions with values in $X$ with the sequentially completed tensor product $\Dcal'\widehat{\otimes}X$.

The notion of vector-valued distribution enables us to formulate the classical field theory involving anticommuting fields in a mathematically elegant way. Note that the map (\ref{d1}) can be treated as an element of ${\frakgo}'(M)\widehat{\otimes}\A(M)$, i.e. a distribution with values in a graded algebra. One can generalize all well known operations like convolution, Fourier transform and pullback to such objects \cite{Hoer,Sch1,Sch2}. For the simplicity of notation we provide here definitions for the case $\Dcal'\widehat{\otimes}\A$.
\begin{df}
Let  $T=t\otimes f$ and $\phi=\varphi\otimes g$, where $f,g\in \A$, $t\in\Dcal'$ and $\varphi\in\Dcal$. We have an antisymmetric bilinear product on $\A$ defined as: $m_a(T,S)\doteq T\wedge S$. We define the convolution of $T$ and $\phi$ by setting:
\begin{equation}
(T*\phi)(x)\doteq T(\phi(x-.))\otimes m_a(f,g)\,.
\end{equation}
The extension by the sequential continuity to $\Dcal'\hat{\otimes}\A$ defines a convolution of a vector-valued distribution with a vector-valued function.
\end{df}
\begin{df}
Let  $T=t\otimes f$ and $S=s\otimes g$, where $f,g\in \A$, $t\in\Ecal'(\RR^2)$ and $s\in\Dcal'$. We define the convolution of $T$ and $S$ by setting:
\begin{equation}
T*S\doteq \int t(.,y)s(y)dy\otimes m_a(f,g)\,,
\end{equation}
This expression is well defined by \cite[4.2.2]{Hoer} and can be extended by continuity to arbitrary $S\in\Dcal'\hat{\otimes}\A$,  $T\in\Ecal'\hat{\otimes}\A$.
\end{df}
\begin{df}
In a similar spirit we define the evaluation of $T=t\otimes f$ on $\phi=\varphi\otimes g$,  by:
\begin{equation}
\left<T,\phi\right>\doteq \left<t,\varphi\right>\otimes m_a(f,g)\,,
\end{equation}
where $f,g\in \A$, $t\in\Dcal'$ and $\varphi\in\Dcal$. Also this can be extended by continuity to  $\Dcal'\hat{\otimes}\A$.
\end{df}
Let  $\mathscr{S}$ denote the space of rapidly decreasing functions, i.e. such that: $\sup_x|x^\beta\partial^\alpha\phi(x)|<~\infty$ for all multi-indices $\alpha,\beta$.
\begin{df}
Let $T\in\mathscr{S}'\hat{\otimes}\A$. We define $\hat{T}\in\mathscr{S}'\hat{\otimes}\A$, the Fourier transform of $T$ as:
\begin{equation}
\hat{T}(\phi)=T(\hat\phi)\qquad\phi\in\mathscr{S}\,.
\end{equation}
\end{df}
Also the notion of the wave front set \cite{Hoer} can be extended to distributions with values in a lcvs. The case of Banach spaces was already treated in detail in \cite{Ko}.
\begin{df}
Let $\{p_\alpha\}_{\alpha\in A}$ be the family of seminorms generating the locally convex topology on $\A$. Let  $T\in\mathscr{S}'\hat{\otimes}\A$. A point $(x,\xi_0)\in T^*\RR^{n}\setminus 0$ is not
in $\textrm{WF}(T)$, if and only if $p_\alpha(\widehat{\phi u}(\xi))$
is fast decreasing as $|\xi|\rightarrow\infty$ for all $\xi$ in an open
conical neighbourhood of $\xi_0$, for some $\phi\in \Dcal$
with $\phi(x)\neq 0$, $\forall \alpha\in A$.
\end{df}
With the notion of the wave front set we can define a "pointwise product" of two distributions $T,S\in\Dcal'\hat{\otimes}\A$ by a straightforward extension of \cite[8.2.10]{Hoer}:
\begin{prop}
Let $T,S\in\Dcal'\hat{\otimes}\A$, $U\in M$ (open). The product $T\cdot S$ can be defined as the pullback of $m_a\circ(T\otimes S)$ by the diagonal map $\delta:U\rightarrow U\times U$ unless $(x,\xi)\in\textrm{WF}(T)$ and $(x,-\xi)\in\textrm{WF}(S)$ for some $(x,\xi)$.
\end{prop}
Obviously we have:  $T\cdot S=-S\cdot T$, whenever these expressions are well defined. In the paper we also use a more suggestive notation: $T\cdot S\doteq\left<T,S\right>$.

\end{document}